\documentclass[a4paper,twocolumn,11pt, accepted=2026-10-23]{quantumarticle}
\pdfoutput=1
\usepackage[utf8]{inputenc}
\usepackage[english]{babel}
\usepackage[T1]{fontenc}
\usepackage{tikz}
\usetikzlibrary{decorations.pathreplacing}

\usepackage{amsmath}
\usepackage[colorlinks=true, linkcolor=blue, citecolor=blue, urlcolor=blue, anchorcolor=black]{hyperref}
\usepackage{orcidlink}
\usepackage{tikz}
\usepackage{lipsum}
\usepackage{amsthm}
\newtheorem{theorem}{Theorem}
\usepackage{subcaption}
\usepackage{comment}
\usepackage{amsfonts}
\usepackage{diagbox}
\usepackage{color}
\usepackage{dsfont}
\usepackage{graphicx}
\usepackage{amssymb}
\usepackage{amsmath,amssymb,bm,amsfonts,mathrsfs,bbm}
\usepackage{physics}

\newtheorem{proposition}{Proposition}
\newtheorem{corollary}{Corollary}
\newtheorem{definition}{Definition}
\newtheorem{lemma}{Lemma}
\newtheorem{observation}{Observation}
\newtheorem{conjecture}{Conjecture}
\usepackage[numbers]{natbib}

\newcommand{\one}[0]{\mathds{1}}
\usetikzlibrary{shapes}
\usepackage{tikz-3dplot}
\usetikzlibrary{intersections, backgrounds}
\usepackage{soul,xcolor}


\begin{document}

\title{
Bounding the entanglement of a state from its spectrum}

\author{Jofre Abellanet-Vidal$^{\orcidlink{0009-0007-3118-0883}}$}
\email{jofre.abellanet@uab.cat}
\affiliation{F\'isica Te\`orica: Informaci\'o i Fen\`omens Qu\`antics, Departament de F\'isica, Universitat Aut\`onoma de Barcelona, E-08193 Bellaterra, Spain.}

\author{Guillem M\"uller-Rigat$^{\orcidlink{0000-0003-0589-7956}}$}
\email{guillem.muller@uj.edu.pl}
\affiliation{Faculty of Physics, Astronomy and Applied Computer Science, Jagiellonian University, ul. \L ojasiewicza 11, 30-348 Kraków, Poland.}

\author{Albert Rico$^{\orcidlink{0000-0001-8211-499X}}$}
\email{albert.ricoandres@uni-siegen.de}
\affiliation{F\'isica Te\`orica: Informaci\'o i Fen\`omens Qu\`antics, Departament de F\'isica, Universitat Aut\`onoma de Barcelona, E-08193 Bellaterra, Spain.}
\affiliation{Naturwissenschaftlich-Technische Fakult\"{a}t, Universit\"{a}t Siegen, Walter-Flex-Stra\ss e 3, 57068 Siegen, Germany.}

\author{Anna Sanpera$^{\orcidlink{0000-0002-8970-6127}}$}
\email{anna.sanpera@uab.cat}
\affiliation{F\'isica Te\`orica: Informaci\'o i Fen\`omens Qu\`antics, Departament de F\'isica, Universitat Aut\`onoma de Barcelona, E-08193 Bellaterra, Spain.}
\affiliation{ICREA -- Instituci\'o Catalana de Recerca i Estudis Avan\c{c}ats, Lluis Companys 23, 08010 Barcelona, Spain.}

\begin{abstract}
We introduce a framework to upper bound  the entanglement content of a bipartite quantum state from its spectrum alone. Using linear maps and their inverses, we derive rigorous constraints on the maximal entanglement that can be activated under global unitary transformations. We  use as entanglement quantifiers the negativity and the Schmidt number; however, our framework is general and applies to any other entanglement measure. Our approach yields compact analytical sufficient criteria for bounding the entanglement of full-rank states in arbitrary dimensions and reveals new spectral constraints on Schmidt number witnesses.
\end{abstract}
\maketitle

\section{Introduction}
\label{sec:introduction}
Recent works have investigated the spectral structure of separable (non-entangled) mixed states whose separability is preserved under arbitrary global unitary transformations, i.e., states that cannot be made entangled in any basis~\cite{kus_geometry_2001}. This raises a complementary question: for those states in which entanglement can be unitarily generated, how much can it be activated given a fixed spectrum? While any separable pure state can be transformed into a maximally entangled state by a suitable global unitary, this is no longer true for sufficiently mixed states. Mixedness, therefore, imposes intrinsic  basis-independent limitations on the entanglement that can be generated.

Detecting entanglement in arbitrary mixed states is notoriously difficult. Major advances in this respect have been achieved through entanglement witnesses, positive maps, and semidefinite-programming hierarchies~\cite{terhal_bell_2000, lewenstein_optimization_2000, guhne_entanglement_2009, doherty_complete_2004, doherty_detecting_2005}. These methods certifying the presence of entanglement, provide lower bounds on the entanglement content of a state and are primarily designed when the full knowledge of the state is accessible. Converse criteria, capable of providing upper bounds on the entanglement compatible with partial information about the state, remain much less developed.

Here, we address this problem and ask how much entanglement can be generated unitarily given a state with a fixed spectrum. Equivalently, we upper bound the maximal entanglement that can be generated from a fixed state by arbitrary global unitary transformations. This spectral formulation is motivated both practically and operationally. From a practical standpoint, spectral properties can be estimated more efficiently than full quantum-state tomography or shadow tomography~\cite{keyl_estimating_2001, lloyd_quantum_2014,Haah2017SampleTomography, cieslinski_analysing_2024}. From an operational standpoint, the problem captures the ultimate entangling potential of a mixed state under the dynamics implemented in quantum circuits. While local optima under unitaries may be obtained with standard numerical techniques, we address the problem of upper bounding global optima.

The entangling power of unitary operations has been extensively studied~\cite{zanardi_entangling_2000, chen_entanglement_2016, linowski_entangling_2020}. However, results that quantify the maximal entanglement attainable along a global unitary orbit are, to the best of our knowledge, largely restricted to two-qubit systems~\cite{munro_maximizing_2001, ishizaka_maximally_2000, verstraete_maximally_2001, wei_maximal_2003} and to qubit states in the symmetric subspace~\cite{serrano-ensastiga_maximum_2023, serrano-ensastiga_entangling_2025}. Existing bounds on entanglement are also often most effective for low-rank states~\cite{liu_bounding_2024, johnston_complete_2022}. By contrast, the regime of full-rank and highly mixed states, where spectral constraints are most relevant is largely unexplored.

We present here a constructive framework for deriving spectral upper bounds on entanglement (see Fig.~\ref{fig:Framework}). Using linear maps and their inverses, we derive powerful sufficient criteria that rule out the possibility that a state with a given spectrum exceeds a prescribed amount of entanglement. We illustrate our approach for two entanglement quantifiers: the negativity of partial transposition~\cite{zyczkowski_volume_1998,vidal_computable_2002}, and the Schmidt number~\cite{terhal_schmidt_2000}.

Before proceeding further, we briefly summarize our main results. First, we show that upper bounding the entanglement attainable under global unitary transformations does not necessarily require full knowledge of the spectrum: in many cases, a single eigenvalue already suffices to do the task. Second, our criteria are effective for highly mixed states, where entanglement characterization is harder. 
Third, our framework, based on invertible linear maps, applies to general entanglement measures beyond the negativity and the Schmidt number. Fourth, as a complementary result, we derive new spectral constraints on Schmidt-number witnesses. \\

The manuscript is organized as follows. In Section~\ref{sec:prelims}, we briefly review the notions of entanglement and entanglement measures for pure and mixed states, and present the main two entanglement measures relevant in our work; the so-called negativity and the Schmidt number. We also introduce linear invertible maps, used to infer entanglement properties directly from the spectrum of a quantum state. In  Section~\ref{sec:AbsEnt} we define the sets of states whose entanglement is upper bounded from spectrum and we present the methods used to characterize them. In Section~\ref{sec:NEGFS}, we characterize the set of states whose negativity is bounded from their spectrum. In Section~\ref{sec:SNFS}, we move to the sets of states whose Schmidt number is constrained by the spectrum.  Additionally, we derive relations connecting these two previously introduced, nonequivalent types of sets, and discuss the spectral properties of our entanglement witnesses. Finally, we present our conclusions in Section~\ref{conclusions}.

During the completion of this work, we became aware of a complementary approach to spectral Schmidt number characterization~\cite{mallick_absolute_2026}.

\section{Preliminaries}
\label{sec:prelims}
Previous studies have focused on characterizing the set of mixed states that remain separable under any unitary transformations~\cite{zyczkowski_volume_1998,kus_geometry_2001, verstraete_maximally_2001, gurvits_separable_2003,johnston_separability_2013,abellanet-vidal_sufficient_2025}, as well as those that remain positive under partial transposition (PPT) under all unitary operations~\cite{hildebrand_PPT_2007,arunachalam_is_2015}. These sets are known as absolutely separable (AS)~\cite{kus_geometry_2001} and absolutely PPT (APPT), respectively. Since global unitaries do not change the spectrum of a state, these sets are also referred to as separable from spectrum (SEPFS) and PPT from spectrum (PPTFS), respectively. In this work, we adopt the latter terminology.
 It is known that the set of SEPFS in $\mathbb C^N\otimes \mathbb C^M$ corresponds to highly mixed states of full rank, with the only exception of the states whose spectrum is  $\lambda_{i}= 1/(D-1)$ for $i=1,\dots, D-1$ and $\lambda_{D}=0$ \cite{verstraete_maximally_2001, zyczkowski_volume_1998,gurvits_separable_2003,johnston_separability_2013,abellanet-vidal_sufficient_2025}. As expected, such states are close to the maximally mixed state (MMS), which is by definition separable and invariant under global unitary transformations.

While the PPTFS sets admit a complete, although complex, characterization based on linear matrix inequalities for arbitrary local dimensions~\cite{hildebrand_PPT_2007}, the complete characterization of SEPFS sets remains unknown in arbitrary dimensions. Furthermore,  it remains an open question whether these two sets coincide \cite{arunachalam_is_2015}. 
Here, by contrast, we focus on characterizing the set of states whose entanglement content remains bounded under arbitrary global unitaries. For completeness, we briefly review some well-known concepts related to bipartite entanglement in pure and mixed states.\\

\noindent\textbf{Pure state entanglement measures.} A pure bipartite quantum state, $\ket{\psi}_{AB}\in\mathcal{H}=\mathbb{C}^{N}\otimes\mathbb{C}^{M}$,  ($N\leq M$) is separable if it can be written as a product state, $\ket{\psi}_{AB}=\ket{\phi}_A\otimes\ket{\varphi}_B$, and it is entangled otherwise. As a consequence of the singular value decomposition, the entanglement of bipartite pure states is fully determined by its Schmidt decomposition: 
\begin{definition}
\label{def:SR}
The Schmidt decomposition of a bipartite pure state $\ket{\psi}\in \mathbb{C}^{N}\otimes\mathbb{C}^{M}$ with $N\leq M$ is given by
\begin{equation}
\label{eq:SchmidtDecomposition}
 \ket{\psi} = \sum_{i=1}^\chi a_i\ket{v_i}\ket{w_i},
\end{equation}
where both $\{\ket{v_i}\}$ and $\{\ket{w_i}\}$ form orthonormal bases. Here $\{a_i>0\}$  are known as Schmidt coefficients, satisfying $\sum_{i=1}^\chi a_i^2=1$, where $\chi$ is the so-called Schmidt rank (SR), with $1\leq\chi\leq N$. 
\end{definition}
Notice that, up to LOCC, the Schmidt decomposition of a state can be written as $\ket{\psi} = \sum_{i=1}^\chi a_i\ket{i}\ket{i}$ without changing its entanglement content. Clearly, a state $\ket{\Psi}$ is separable if and only if its Schmidt rank is one ($\chi = 1$). The Schmidt coefficients can be conveniently arranged in a vector $\boldsymbol{a}=(a_1^2,...,a_\chi^2)$, with $a_i\geq a_{i+1}$. Entanglement transformations between pure states under local operations and classical communication (LOCC) are given by Nielsen majorization criterion \cite{nielsen_conditions_1999}: a state $\ket{\psi}_{AB}$ with Schmidt vector $\boldsymbol{a}$ can be transformed via LOCC  into a state $\ket{\phi}_{AB}$ with Schmidt vector $\boldsymbol{b}$ if and only if $\boldsymbol{b}$ majorizes $\boldsymbol{a}$, written $\boldsymbol{b}\succ\boldsymbol{a}$, namely
\begin{equation}
    \sum_{i=1}^k b_i^2\geq \sum_{i=1}^k a_i^2,
\end{equation}
for all integers $k=1,\dots\chi$ indexing the Schmidt coefficients. Therefore, the Schmidt decomposition fully characterizes the entanglement properties of a bipartite pure state. A complete set of entanglement measures is given by the partial sums
\begin{equation}
\label{eq:PartialSums}
    m_k := 1-\sum_{i=1}^k a_i^2\,.
\end{equation}
From a resource-theory perspective, entanglement is a resource while separable states are free states. Since monotonicity of entanglement under local operations and classical communication (LOCC) is considered the only natural requirement which an entanglement measure should fulfill,  meaning that entanglement cannot be increased nor be generated by LOCC, any quantity $\text{EM}(\ket{\psi})$ that satisfies this condition naturally vanishes for separable states and defines an entanglement measure (EM)  \cite{guhne_entanglement_2009,horodecki_quantum_2009, eltschka_quantifying_2014, plenio_introduction_2007, vedral_quantifying_1997, vidal_entanglement_2000}. The Schmidt rank $\chi$ of a state or its corresponding partial sums $\{m_k\}$ are examples of entanglement measures for bipartite pure states.   \\

\begin{figure}
    \centering
    \includegraphics[width=\linewidth]{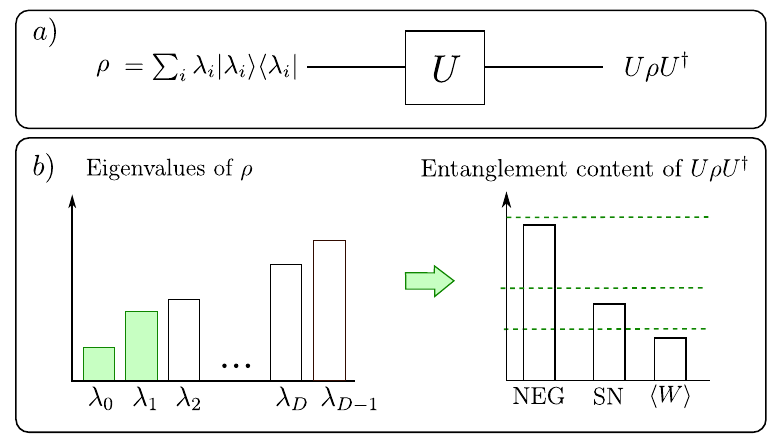}
    \caption{ $a)$ Given a noisy mixed quantum state $\rho$, we aim to address the question of how much entanglement it can have under the action of arbitrary circuit gates $U$. $b)$ From the knowledge of a few eigenvalues of $\rho$, we derive upper bounds to the entanglement content of the output state $U\rho U^\dagger$ with respect to various entanglement quantifiers, such as witness violation (Theorem~\ref{thm:witnessAlb}), entanglement negativity (Section~\ref{sec:NEGFS}), and Schmidt number (Section~\ref{sec:SNFS}).  }
    \label{fig:Framework}
\end{figure}

\noindent \textbf{Mixed state entanglement measures.} Entanglement of mixed states is considerably more difficult to characterize than entanglement of pure states.
We recall that a  bipartite mixed state, $\rho_{AB}$, acting on  $\mathcal{H}=\mathbb{C}^{N}\otimes\mathbb{C}^{M}$, ($N\leq M$) is separable if it can be written as 
\begin{equation}
\rho_{AB}=\sum_{i}p_{i}\ketbra{\phi_{i}}{\phi_{i}}_A\otimes\ketbra{\varphi_{i}}{\varphi_i}_B,
\end{equation}
with $p_i\geq 0$ and $\sum_ip_i=1$. Otherwise, it is entangled. From now on, for simplicity, we will omit the subsystem subscripts $AB$.  Since the decomposition of a density matrix as a convex combination of pure states $\{p_{i},\ket{\psi_{i}}\}$ is not unique, defining entanglement measures for mixed states is much harder. Given an entanglement measure $\text{EM}$ for pure states, it can be extended to mixed states as
\begin{equation}
\label{eq:MixedMeasure}
\text{EM}(\rho)=\inf_{\{p_i,\ket{\psi_i}\}}\max_i \text{EM}(\ket{\psi_i}),
\end{equation}
where the minimization is taken over all possible ensembles decomposing $\rho$ and the maximization is taken over all elements of the specific decomposition. However, the above minimization is a difficult task in general. The convex roof extension  \cite{uhlmann_roofs_2010} of pure state entanglement (convex) monotones, also allows to define mixed state entanglement measures:
\begin{equation}
\label{eq:MixedMonotone}
\text{EM}(\rho)=\min_{\{p_i,\ket{\psi_i}\}}\sum_{i}p_{i}\text{EM}(\ket{\psi_i}),
\end{equation}
where, again, the minimization is taken over all possible ensembles decomposing $\rho$. A key property of both, Eqs.~\eqref{eq:MixedMeasure} and \eqref{eq:MixedMonotone}, is that they are convex functions by construction, which guarantees that the entanglement content of a state cannot increase under convex mixtures. 

In particular, here we consider two well known bipartite EM: negativity and
Schmidt number.  Negativity (NEG)~\cite{zyczkowski_volume_1998, vidal_computable_2002} measures the action of a map that is positive (P) but not completely positive (CP); the transposition map \cite{peres_separability_1996, horodecki_separability_1996}. It does not arise as an extension of a pure state EM, even though convexity allows to upper bound it from pure states. It is defined as: 
\begin{definition}
\label{def:Negativity}
The negativity $\mathcal{N}$ of a bipartite mixed state $\rho \in \mathcal B(\mathbb{C}^{N}\otimes\mathbb{C}^{M})$ is given by 
\begin{equation}
\label{eq:Negativity}
\mathcal{N}(\rho) = \frac{|| \rho^\Gamma||_1-1}{2},
\end{equation}
where $||A||_1 = \mathrm{Tr}(\sqrt{A^\dagger A})$ is the trace norm and $\Gamma$ denotes the partial transpose of $\rho$ with respect to any of the subsystems.
\end{definition}
The negativity can also be extended as a convex roof measure of the form of Eq.~\ref{eq:MixedMonotone}, known as convex roof extended negativity $\mathcal{N}^{\text{CREN}}$ \cite{lee_convex-roof_2003}. 
\begin{equation}
\label{eq:NegativityCREN}
\mathcal{N}^{\text{CREN}}(\rho) = \min_{\{p_j, |\psi_j\rangle\}} \sum_j p_j \, \mathcal{N}\left(|\psi_j\rangle\langle\psi_j|\right) 
\end{equation}
It is known that, for mixed states $\mathcal{N}\leq\mathcal{N}^{\text{CREN}}$.

Similarly, any P but not CP map has an associated negativity leading to an EM (see the discussion in Appendix~\ref{sec:NegativityReduction} for the case of the reduction map). Another important measure of entanglement is the  so-called Schmidt number (SN). It corresponds to the extension of the Schmidt rank of pure states to mixed states. 
\begin{definition}
\label{def:SN}
 A state $\rho \in \mathcal B(\mathbb{C}^{N}\otimes\mathbb{C}^{M})$ has Schmidt number  $\chi$ if there exists a decomposition $\rho = \sum_i p_i\ketbra{\psi_i}$, with $p_i\geq 0$, where all pure states $\ket{\psi_i}$ have at most Schmidt rank $\chi$, and admits no decomposition where all states have less than $\chi$ Schmidt rank. 
\end{definition}
Notice that the monotones $m_k$ in~\eqref{eq:PartialSums} vanish for $k>\chi$, and thus provide a sense of noise robustness at detecting Schmidt rank and Schmidt number.  
Last, we briefly recall how entanglement of physical systems is detected in practice. To this end, one typically makes use of observables known as entanglement witnesses~\cite{horodecki_separability_1996}:
\begin{definition}
    An entanglement witness $W$ is an observable, $W=W^\dag$, whose expectation value is nonnegative on any separable state $\rho$, $\Tr(\rho W)\geq 0$, and there exists at least one entangled state, $\rho_{ent}$, such that  $\Tr(\rho_{ent} W) <0$
\end{definition}

A bipartite entanglement witness acting on $\mathcal{H}_A\otimes\mathcal{H}_B$ is called decomposable \cite{terhal_family_2001,lewenstein_characterization_2001} if it can be written in the form $W = P+Q^{\Gamma}$ where $P$ and $Q$ are positive semidefinite operators. It follows that decomposable witnesses cannot detect entangled states with positive partial transpose.

\section{Entanglement upper bounds from the spectrum}
\label{sec:AbsEnt}
We start by defining the sets of bipartite states whose entanglement is upper bounded under any global unitary transformation.

\begin{definition}
\label{def:EFS}
A state $\rho \in \mathcal{B}(\mathbb{C}^{N} \otimes \mathbb{C}^{M})$ is said to be $\gamma$-entangled from spectrum, $\mathrm{EFS}_{\gamma}$, w.r.t. the entanglement measure EM, if for every global unitary $U$, it holds that
\begin{equation}\label{eq:defEFSgamma}
\mathrm{EM}(U \rho U^{\dagger}) \leq \gamma,
\end{equation}
and the bound is saturated for at least one such unitary.
\end{definition}

\noindent 
Definition~\ref{def:EFS} considers the specific property of the state. Nevertheless, to ensure convexity of the resulting sets, we define the $\text{EFS}_\gamma$ sets disregarding the attainability condition, in a similar approach as in Ref.~\cite{terhal_schmidt_2000}. Then, they satisfy $\mathrm{EFS}_{\gamma'}\subset \mathrm{EFS}_\gamma$ for $\gamma'<\gamma$. Note also that there is not a lower bound on entanglement, 
as there always exists a global unitary, 
$ U = \sum_{i,j}\ketbra{ij}{\lambda_{f(i,j)}}$,   where $\{\ket{i,j}\}$ denotes the computational basis  and $f(i,j)$  is a proper function such that: 
$$\tilde\rho= U \left(\sum_{\iota}\lambda_{\iota}\ketbra{\lambda_{\iota}}{\lambda_{\iota}}\right) U^{\dagger}=
\sum_{\iota}\lambda_{\iota}\ketbra{ij}{ij},$$
is a separable state ($\gamma = 0$) for any EM \cite{clarisse_disentangling_2007}.  Clearly, $\mathrm{SEPFS}\subset \text{EFS}_{\gamma}$ for any $\gamma >0$, since the SEPFS states remain separable for any global unitary. 

To better illustrate the problem under consideration, in Fig.~\ref{fig:Pictorial} we provide a sketch of the different nested convex sets characterizing the entanglement content of quantum states according to a given EM, and the corresponding nested convex subsets $\mathrm{EFS}_{\gamma}$ representing the quantum states whose entanglement cannot increase under global unitary transformations. 
 \begin{figure}
    \centering
    \includegraphics[width=\linewidth]{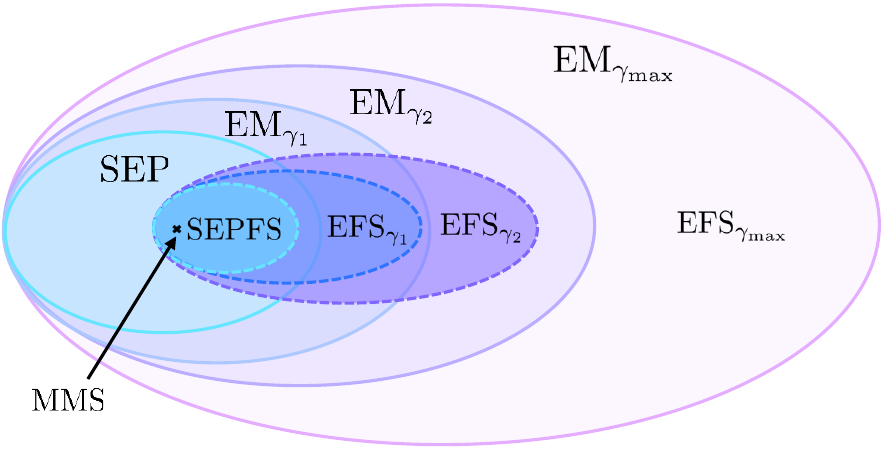}
    \caption{Pictorial representation of the set of quantum states bounded by a convex entanglement measure ($\text{EM}(\rho)\leq\gamma$) and of the sets with $\text{EFS}_{\gamma}$ , considering $0<\gamma_{1}<\gamma_{2}<\gamma_{\max}$. Note that when $\gamma=\gamma_{\max}$, the whole set of states is recovered. The maximally mixed state is denoted as MMS.}
    \label{fig:Pictorial}
\end{figure}
\subsection{Characterization of the $\text{EFS}_{\gamma}$ sets}
\label{sec:HowTo}
In order to characterize the $\text{EFS}_{\gamma}$ sets, we reformulate the techniques presented in \cite{lewenstein_sufficient_2016,abellanet-vidal_sufficient_2025, lewenstein_linear_2022}, which use families of positive invertible linear maps to provide sufficient conditions to certify SEPFS. Here, the same technique is used to certify if $\rho\in\text{EFS}_{\gamma}$. This method, however, does not allow the complete characterization of the $\text{EFS}_{\gamma}$ sets. 

Our approach is built upon Theorem~\ref{thm:InverseMapSNFS}, which makes use of the unitarily covariant reduction map~\cite{bardet_characterization_2020, bhat_linear_2011}
\begin{equation}
\label{eq:ReductionMap}
\Lambda_{\alpha}(\sigma)=\Tr(\sigma)\cdot\mathds{1}+\alpha\sigma,
\end{equation}
and its unitarily covariant inverse
\begin{equation}
\label{eq:ReductioInversenMap}
\Lambda_{\alpha}^{-1}(\sigma)=\frac{1}{\alpha}\left(\sigma-\frac{\Tr(\sigma)\mathds{1}}{D+\alpha} \right).
\end{equation}
We keep the unnormalized notation because the positivity of $\Lambda_\alpha^{-1}$, the spectral inequalities, and the associated convex-geometric construction are invariant under multiplication by a positive scalar.

\begin{theorem}
\label{thm:InverseMapSNFS}
Let $\Lambda_{\alpha}$ be the family of the reduction maps. By $\alpha_{-},\alpha_{+}$ we denote the range of the parameter $\alpha\in[\alpha_{-},\alpha_{+}]$, for which every $\rho\in\mathcal{B}(\mathbb{C}^{N}\otimes\mathbb{C}^{M})$ fulfills $\text{EM}[\Lambda_{\alpha}(\rho)]= EM[\sigma]\leq\gamma$. Then, if $\Lambda_{\alpha}^{-1}(\sigma) \geq 0$, 
it follows that $\sigma\in\text{EFS}_{\gamma}$.  \end{theorem}
\begin{proof} (See also \cite{lewenstein_sufficient_2016}). By considering the convex sets $\text{EFS}_{\gamma}$,
    if $\Lambda_{\alpha}^{-1}(\sigma) \geq 0$, then $\Lambda_{\alpha}[\Lambda_{\alpha}^{-1}(\sigma)]=\sigma$.
    \end{proof}
    
Since the reduction map is a unitarily covariant linear map, i.e.,  $\Lambda_{\alpha} (U\rho U^{\dagger})= U'\Lambda(\rho)U'^{\dagger}$, with $U,U'$ unitaries, the condition $\Lambda_{\alpha}^{-1}(\sigma) \geq 0$ depends only on the eigenvalues of $\sigma$ and is invariant under global unitary transformations~\cite{abellanet-vidal_sufficient_2025}. Thus, the characterization of each set, $\text{EFS}_{\gamma}$, requires, according to Theorem \ref{thm:InverseMapSNFS}, to determine the corresponding values of $\alpha_{\mp}$. Decisively, $\alpha_{-} = -1$ is the extremal threshold ensuring the positivity of the reduction map, as shown previously in the context of $\mathrm{SEPFS}$~\cite{lewenstein_sufficient_2016}. This implies that $\alpha_{-} = -1$ is also the extreme value for all sets $\mathrm{EFS}_{\gamma}$, since $\mathrm{SEPFS} \subset \mathrm{EFS}_{\gamma}$ for any $\gamma > 0$ and any EM. In contrast, the value of $\alpha_{+}$ generally depends on both the specific EM under consideration and the chosen threshold $\gamma$.
To determine the value of $\alpha_{+}$, we make use of Lemma~\ref{Lemma:GeneralAlphaPM} (originally derived in~\cite{abellanet-vidal_sufficient_2025}),  to obtain a linear condition on the eigenvalues of $\rho$ certifying that $\rho \in \mathrm{EFS}_{\gamma}$. This provides an inner characterization of the desired sets. Moreover, since the reduction map is linear, it suffices to analyze its action on pure states. For this purpose, we focus on the so-called pseudo-pure states (PPS), i.e.,
\begin{equation}
\rho_{PPS} = (1-p)\ketbra{\psi}{\psi} + p\frac{\mathds{1}_{D}}{D},
\end{equation}
where the noise parameter is given by $p = D/(D+\alpha_{+})$. In Appendix~\ref{sec:PPS}, we provide an operational way to infer the value $p$. 

\begin{lemma}{\cite{abellanet-vidal_sufficient_2025}}
\label{Lemma:GeneralAlphaPM} 
 Let $\rho \in \mathcal B(\mathbb{C}^{N}\otimes\mathbb{C}^{M})$  be a normalized state acting on a Hilbert space of global dimension $D=NM$, and $\boldsymbol{\lambda} = \{\lambda_i^{\uparrow} \}_{i=0}^{D-1}$ the corresponding eigenvalues in a non-decreasing order, with $\sum_{i=0}^{D-1}\lambda_i = 1$.  Given $\alpha_{+}\geq 0$ and $\alpha_{-}\leq0$, the vector $\boldsymbol{\lambda}$ is contained in the convex hull defined by the conditions 
\begin{equation}
\label{eq:DemoLemma1}
\lambda_{0}\geq \frac{1}{D+\alpha_{+}},  
\end{equation}
\begin{equation}
\label{eq:DemoLemma2}
\lambda_{D-1}\leq \frac{1}{D+\alpha_{-}},   
\end{equation} 
if and only if
 \begin{equation}
 \label{eq:GeneralAlphaPM}
K \cdot \sum_{i=0}^{c-1} \lambda_i + \left[D -K \cdot c + \alpha_+ \right] \cdot \lambda_c \geq 1,
\end{equation}
where $ K= \left( 1 - \frac{\alpha_+}{\alpha_-} \right)$, and 
\begin{equation*}
c = \left\lceil \frac{\alpha_+ + \alpha_- (D - 1 + \alpha_+)}{\alpha_- - \alpha_+} \right\rceil.
\end{equation*}
\end{lemma}

Specifically, if we consider $\alpha_{-}=-1$ from positivity of the map, the whole problem is reduced to calculating $\alpha_{+}$ since we can express the whole equation as a function of only $\alpha_{+}$ and $D$ as
\begin{equation*}
\label{eq:implications}
K=1+\alpha_{+},\quad c=\lceil \frac{D-1}{1+\alpha_{+}} \rceil.
\end{equation*}

Lemma~\ref{Lemma:GeneralAlphaPM} exploits the results obtained for PPS based on $\alpha_{+}$ and $\alpha_{-}$ to construct sufficient conditions for membership in $\mathrm{EFS}_{\gamma}$ for arbitrary mixed states. As a result, this approach provides simple analytical conditions that certify when the entanglement content (according to a given entanglement measure) of a mixed state cannot exceed $\gamma$ under any global unitary transformation. As a byproduct, only a few eigenvalues are required to certify that $\rho \in \mathrm{EFS}_{\gamma}$, rather than the full spectrum. Moreover, since the eigenvalues are non-decreasingly ordered, it is possible to substitute $\lambda_{c}\geq\lambda_{c-1}$ in Eq.~\eqref{eq:GeneralAlphaPM} 
The resulting condition detects fewer states, but also requires knowledge of fewer eigenvalues to certify belonging to the $\text{EFS}_{\gamma}$ sets.\\ 

Lemma~\ref{Lemma:GeneralAlphaPM} can be applied in two complementary directions: either to determine the maximal value of $\gamma$ compatible with a prescribed spectrum, or to identify the spectra compatible with a prescribed entanglement threshold $\gamma$. In Fig.~\ref{fig:SetsLemma1}, we provide an intuitive geometric interpretation of our criteria using the eigenvalues of the states under examination.

\begin{figure}[h!]

\centering
\begin{tikzpicture}[scale=3.5]
\definecolor{asep1}{HTML}{66E5FF}
\definecolor{asep2}{HTML}{000000}
\definecolor{asn2}{HTML}{2F74FF}
\definecolor{asn3}{HTML}{7F66FF}
\definecolor{asnN}{HTML}{E4ADFF}

\fill[black, opacity=0.2] 
    (0,0.81649658) -- 
    (-0.35355,0.204) --
    (0,0) -- 
    cycle;

\filldraw[fill=asnN,draw=asnN, very thick, solid, fill opacity = 0.1] (0,0.81649658) -- (-0.70710678,-0.40824829) -- (0.70710678,-0.40824829) -- cycle;
    \filldraw[fill=asnN,draw=asnN] (0,0.81649658) circle (0.7pt);
    \filldraw[fill=asnN,draw=asnN] (-0.70710678,-0.40824829) circle (0.7pt);
    \filldraw[fill=asnN,draw=asnN] (0.70710678,-0.40824829) circle (0.7pt);

\filldraw[fill=asn3, draw=asn3, very thick, dashed, fill opacity = 0.1] (0,0.65) -- (0.3536,0.204) -- (0.5628,-0.324996) -- (0,-0.408) -- (-0.5628,-0.324996) -- (-0.35355,0.204) -- cycle ;

\filldraw[fill=asn2,draw=asn2,very thick, dashed, fill opacity=0.1] (0,0.5) -- (0.3536,0.204) -- (0.4330086,-0.249996) -- (0,-0.408) -- (-0.4330086,-0.249996) -- (-0.35355,0.204) -- cycle ;

\filldraw[fill=asep1, draw=asep1, very thick,dashed, fill opacity=0.2] (0,0.326598) -- (0.3536,0.204) -- (0.2828,-0.1633) -- (0,-0.408) -- (-0.2828,-0.163299) -- (-0.35355,0.204) -- cycle ;

\draw[->] (0,-0.408) -- (0,0.9) node[right] {$\lambda_0$};
\draw[->] (0.3536,0.204) -- (-0.77,-0.445) node[above left] {$\lambda_1$};
\draw[->] (-0.35355,0.204) -- (0.77,-0.445) node[above right] {$\lambda_2$};

\node[scale=0.6] at (0.05,0.05) {MMS};

\end{tikzpicture}
\caption{Geometry of the $\mathrm{EFS}_{\gamma}$ inner characterization using Lemma~\ref{Lemma:GeneralAlphaPM} in the probability simplex described by the eigenvalues $\boldsymbol{\lambda}$ of the density matrix for $D=3$ in barycentric coordinates. The figure is illustrative as $D=3$ does not correspond to any bipartite splitting. $0\leq\gamma_{1}<\gamma_{2}<\gamma_{\max}$ is assumed, following Fig.~\ref{fig:Pictorial}.}
    \label{fig:SetsLemma1}
\end{figure}
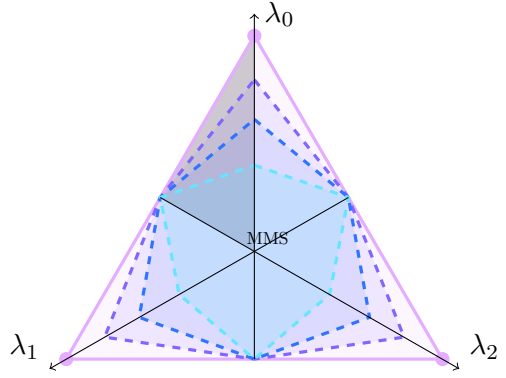

\section{Negativity from spectrum}
\label{sec:NEGFS}

We proceed now to characterize the set of bipartite states whose negativity cannot exceed a given value $\gamma$ under global unitaries. In accordance with Def.~\ref{def:Negativity}, we denote these sets as $\mathrm{NEGFS}_{\gamma}$.

\begin{definition}
\label{def:NEGFS}
A bipartite state $\rho \in \mathcal{B}(\mathbb{C}^{N} \otimes \mathbb{C}^{M})$ has $\mathrm{NEGFS}_{\gamma}$ if, for every global unitary $U$, $\mathcal{N}(U \rho U^{\dagger}) \leq \gamma$ and the bound $\gamma$ is saturated for at least one such unitary. 
\end{definition}

We recall that, as in Definition~\ref{def:EFS}, the sets of states are defined without the attainability condition to ensure convexity. Although similar sets can be defined for other P but not CP maps, the negativity associated to the transposition map yields better conditions to quantify entanglement in arbitrary dimensions (see Appendix~\ref{sec:NegativityReduction}). At the operational level, determining upper bounds on the negativity of a state $\rho$ under partial transposition is motivated by the following observation.

\begin{theorem}
\label{thm:witnessAlb}
Let $W$ be a decomposable entanglement witness ($W^{T_A}>0$), normalized to unit trace. Let $\rho$ be a state detected by this witness, $\Tr(\rho W)<0$. Then, $|\Tr(\rho W)|\leq \mathcal{N}(\rho)$.
\end{theorem}
\begin{proof}
Denote the smallest eigenvalue of $\rho^{T_A}$ as $\lambda$, i.e. $\lambda:=\lambda_{\min}(\rho^{T_A})$, and its corresponding eigenvector as $\ket{\lambda}$. It holds that $\lambda=\Tr(\rho^{T_A}\ketbra{\lambda}{\lambda})=\Tr(\rho\ketbra{\lambda}{\lambda}^{T_A})$, which by definition equals the smallest possible expectation value $\Tr(\rho W)$ over all Hermitian operators $W$ with $W^{T_A}\geq 0$ and $\Tr(W)=1$. The negativity is the absolute value of the sum of negative eigenvalues of $\rho^{T_A}$, and therefore $\mathcal{N}(\rho)\geq|\lambda|$.
\end{proof}
Although this observation is rather common knowledge, we make it explicit because it links the negativity of partial transposition, which is a theoretical value, with witness violation, which is the main tool used in experiments to detect entanglement. Therefore, it motivates the following section from the operational point of view: the bounds derived below for the negativity will in practice also bound the largest possible violation of decomposable entanglement witnesses.

\bigskip

\noindent \textbf{Characterization of the $\text{NEGFS}_{\gamma}$ sets.}
To characterize the $\mathrm{NEGFS}{\gamma}$ sets, we must first determine the corresponding value $\alpha_{+}$ associated to that set. To this end we restrict our analysis to PPS whose negativity is characterized in the following lemma:

\begin{lemma}
\label{lemma:NegativityIsotropic}
Given a normalized pure state $\ket{\psi}=\sum_{i}a_{i}\ket{ii}\in \mathbb{C}^{N}\otimes\mathbb{C}^{M}$, with Schmidt coefficients $a_i$, the negativity of its PPS with $p= D/(D+\alpha)$ is given by
\begin{equation}
\label{eq:NegativityIsotropic} 
\mathcal{N}(\rho)=\frac{1}{D+\alpha}\sum^{\chi}_{\substack{i>j \\ 
a_ia_j>\alpha^{-1}}}\big( \alpha a_{i}a_{j}-1\big). 
\end{equation}
\end{lemma}

\begin{proof}
Due to the linearity of partial transposition,  $\mathcal{N}(\rho)$ depends only on the negative eigenvalues of $\ketbra{\psi}{\psi}^{\Gamma}$. Using the Schmidt form: $\ket{\psi} = \sum_{i=1}^{\chi} a_i\ket{ii}$, $\ketbra{\psi}^{\Gamma} = \bigoplus_{i=1}^\chi a_i^2\ketbra{ii}\bigoplus_{i>j}a_ia_j(\ketbra{ij}{ji}+\ketbra{ji}{ij})$, we obtain $\mathcal{N}(\ketbra{\psi}) = \sum_{i>j}a_ia_j= (\sum_ia_i)^2/2-1/2$. 
\end{proof}

We illustrate Lemma~\ref{lemma:NegativityIsotropic} in Fig.~\ref{fig:Negativities}, where we display the maximal attainable negativity of a PPS as a function of the noise parameter $p$, for several pure states $\{\ket{\psi_i}\} \in \mathbb{C}^{6} \otimes \mathbb{C}^{6}$ with Schmidt ranks $\chi_i = 2, \dots, 6$ and uniform vector of Schmidt coefficients. We observe that the largest compatible product of Schmidt coefficients, $a_i a_j$, occurs when $\chi = 2$. Therefore, although large negativity is achieved only for states with maximal Schmidt rank, PPS with smaller Schmidt rank are more robust to noise (see the inset of the figure), as first noted in Ref.~\cite{vidal_robustness_1999}. Moreover, we find that if
\begin{equation}
\label{eq:pchi}
    2(\chi - 1) \leq \alpha_{+} \leq 2\chi,
\end{equation}
then the most entangled PPS compatible with this value of $\alpha_{+}$ comes from a pure state with Schmidt rank $\chi$.

As also shown in Fig.~\ref{fig:Negativities}, Haar-random unitaries are highly entangling, since they transform diagonal spectra into states with large negativity. Nonetheless, our numerical simulations~\cite{qutip_random_4_0_2} indicate that most Haar-random unitaries achieve negativity values significantly below the absolute maximum~\cite{datta_negativity_2010}.

With the above results for PPS, we proceed now to derive the conditions on the spectra needed to ensure that $\rho\in NEGFS_{\gamma}$.

\begin{theorem}
\label{thm:AlphaNegativity} 
 Let $\rho \in \mathcal B(\mathbb{C}^{N}\otimes\mathbb{C}^{M})$  be a normalized state ($N\leq M$, $D=N M$) with non-decreasingly ordered spectrum $\boldsymbol{\lambda} = \{\lambda_i^{\uparrow} \}_{i=0}^{D-1}$ 
 and $\sum_{i=0}^{D-1}\lambda_i = 1$. Given $\gamma$, there exists a parameter $\tau\leq N$, 
$\tau\in\mathbb{N}$, such that
\begin{equation}
\label{eq:RangeGammaNegativity}
\frac{\tau(\tau-1)}{2(2\tau+D)}>\gamma\geq\frac{(\tau-1)(\tau-2)}{2(2\tau+D-2)}\quad\text{for } 2\leq \tau < N,
\end{equation}
\begin{equation}
\label{eq:RangeGammaNegativity2}
\frac{N-1}{2}\geq\gamma\geq\frac{(N-1)(N-2)}{2(D+2N-2)}\quad\text{for } \tau = N.
\end{equation}
Then, fixing $\alpha_{+}=\frac{\tau^2-\tau+2D\gamma}{\tau-2\gamma-1}$, if
\begin{eqnarray}
&\lambda_{0}&\geq \frac{1}{D+\alpha_{+}}, \quad\text{or}\\
 \label{eq:GeneralAlphaPMNeg}
&K&\cdot \sum_{i=0}^{c-1} \lambda_i + \left[D -K \cdot c + \alpha_+ \right] \cdot \lambda_c \geq 1
\end{eqnarray}
then $\rho\in \text{NEGFS}_{\gamma}$, where $ K= \left( 1 + \alpha_+ \right)$, and $c = \left\lceil \frac{D-1}{1+\alpha_+} \right\rceil$.
\end{theorem}
\begin{proof}
Given a pure state of Schmidt rank $\chi$, $\ket{\psi_\chi}$, the image of the reduction map is $\Lambda_\alpha(\ketbra{\psi_\chi}{\psi_\chi}) = \alpha\ketbra{\psi_\chi}{\psi_\chi} + \mathbb{I}:=\rho_\chi$, which corresponds to the unnormalized PPS addressed in Lemma~\ref{lemma:NegativityIsotropic}. From Eq.~\eqref{eq:pchi}, we derive the ranges of $\alpha_{+}$ for which the mixing with $\ket{\psi}$ with SR $\chi$ yields the maximal negativity.
\end{proof}

Theorem~\ref{thm:AlphaNegativity}, which is exact for PPS,  establishes a sufficient condition to determine whether a bipartite state $\rho$ belongs to $\mathrm{NEGFS}_{\gamma}$ based solely on $c = \left\lceil \frac{D-1}{1+\alpha_+}\right\rceil$ eigenvalues. Equivalently, this result identifies the corresponding value of $\alpha_{+}$ for each set $\mathrm{NEGFS}_{\gamma}$ in arbitrary dimensions. It is also relevant to study its asymptotic behavior, $N = M \rightarrow \infty $. We obtain a scaling for the turning points in Eq.~\eqref{eq:pchi} as $p=\frac{D}{D+2(\chi-1)}\propto 1$. Nevertheless, the negativity bounds introduced in Eq.~\eqref{eq:RangeGammaNegativity} $\gamma\propto\frac{\tau^2}{2D}\propto\frac{1}{2}$ for $\tau\sim N$. This indicates that, as the dimension increases, we would be able to find greater negativities for small values of $\alpha$ ($p\rightarrow 1$). 
\begin{figure}
    \centering
    \includegraphics[width=\linewidth]{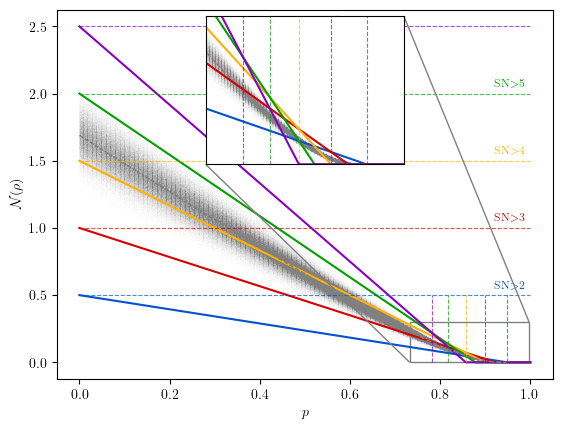}
    \caption{Maximal negativity of PPS of Schmidt rank $\chi$ as a function of $p =D/(D+\alpha)$ for $N = M = 6$. The dashed lines indicate the maximal negativity for pure Schmidt rank $\chi$ states (or convex combinations of them). In black, we depict the negativity obtained for each spectrum, maximizing over a sample of $10^4$ Haar random unitaries.} 
    \label{fig:Negativities}
\end{figure}

This theorem can be used as a criterion to bound the maximal negativity that an arbitrary $\rho$ might have given its eigenvalues $\boldsymbol{\lambda}$. We illustrate this in Fig.~\ref{fig:NEGFSnumerics}, where for simplicity we show the two-qubit case, whose maximal negativity is bounded by $\mathcal{N}(\ketbra{\psi})\leq 0.5$. We display the maximal $\mathcal{N}(\rho)$ obtainable according to Theorem~\ref{thm:AlphaNegativity} for different states (solid lines) and compare them with the maximal negativity obtained under random unitaries (dashed lines).
\begin{figure}[h]
    \centering
    \includegraphics[width=\linewidth]{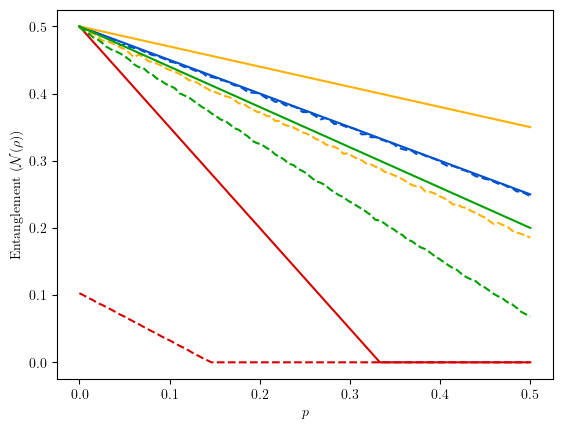}
    \caption{ Maximal negativity obtained from the spectrum with $10^4$ random unitary matrices given different parametrized $2$-qubit spectra (dashed lines) along with our bound on the maximal negativity obtainable for each spectrum (solid line) following Theorem~\ref{thm:AlphaNegativity}. The depicted spectra are of the form $(1-p/2,p/6,p/6,p/6)$ in blue, $((1-p)/2,(1-p)/2,p/2,p/2)$ in red, $(1-19p/30,p/3,p/5,p/10)$ in green and $(1-13p/15,p/3,p/3,p/5)$ in yellow.}
    \label{fig:NEGFSnumerics}
\end{figure}
Importantly,  Theorem~\ref{thm:AlphaNegativity} does not provide an upper bound on the negativity for rank-deficient states as illustrated by the fact that for $p=0$, all solid lines recover the trivial bound on the negativity.  Furthermore, since Theorem~\ref{thm:AlphaNegativity} is exact for PPS, the corresponding upper bound for the blue spectrum is nearly saturated, as it coincides with the PPS family analyzed above. For more general spectra, however, the upper bound becomes looser.

The entanglement negativity is arguably the central bipartite entanglement measure for NPT states, as it provides useful bounds on several demanding entanglement quantifiers, such as the Schmidt number~\cite{eltschka_negativity_2013} or the concurrence~\cite{eltschka_partial_2015}, both of which are considerably more difficult to compute.

\subsection{Bounding a complete family of entanglement monotones}
\label{sec:mks}
We conclude this section by discussing the use of the negativity to bound the cumulative sums of Schmidt coefficients $m_k$ as defined in Eq.~\eqref{eq:PartialSums}, when extended to mixed states  $\rho$ through Eq.~\eqref{eq:MixedMeasure}. These are of particular interest, since they constitute a complete family of monotones that determine LOCC-interconvertibility of the pure states decomposing $\rho$. Similar approaches have been used to bound entanglement measures and monotones from the negativity of positive maps~\cite{eltschka_partial_2015,chen_concurrence_2005,chen_entanglement_2005,fei_r_2006,li_measurable_2010,wootters_entanglement_1998,hill_entanglement_1997}.

\begin{proposition}
Let $\rho\in\mathcal{B}(\mathbb{C}^{N}\otimes\mathbb{C}^{M})$, with $N\leq M$. Then, its negativity $\mathcal{N}(\rho)$ is upper bounded by
\begin{equation}\label{eq:NegBoundMk}
\mathcal{N}(\rho)\leq (kx + (N-k)y)^2/2-1/2\,,
\end{equation}
where $x=\sqrt{(1-m_k)/k}$ and $y=\sqrt{m_k/(N-k)}$.
\end{proposition}
\begin{proof}
1. As $\mathcal{N}(\ketbra{\psi}{\psi})=(\sum_ia_i)^2/2-1/2$, and the negativity of $\rho=\sum_jp_j\ketbra{\psi_j}{\psi_j}$ is convex, i.e., $\mathcal{N}(\rho)\leq \sum_jp_j\mathcal{N}(\ketbra{\psi_j}{\psi_j})$, it follows that $\mathcal{N}(\rho)\leq\ [(\sum_ia_i^*)^2-1]/2$, where $\{a_i^*\}$ are the Schmidt coefficients of the state $\ket{\psi_{j^*}}$ with maximal negativity among all states $\ket{\psi_j}$, for any decomposition of $\rho$. 

2. For a fixed $m_k(\ketbra{\psi_{j^*}}{\psi_{j^*}})=1-\sum_{i=1}^k {a_i^{*}}^2$, the negativity $\mathcal{N}(\ketbra{\psi_{j^*}}{\psi_{j^*}})$ is maximized for $a_1^*=...=a_k^*:=x$ and $a_{k+1}^*=...=a_N^*:=y$, where
$ x = \sqrt{\frac{1-m_k}{k}}$ and  $y=\sqrt{\frac{m_k}{N-k}}$
are fixed due to the normalization constraint $\sum_i {a_i^{*}}^2 = kx^2+(N-k)y^2=1$, leading to Eq.~\eqref{eq:NegBoundMk}.

3. It remains to be shown that the state $\ket{\psi_{j^*}}$ that maximizes $\mathcal{N}(\ketbra{\psi_j}{\psi_j})$ among all pure states in a decomposition of $\rho$ also maximizes the monotone $m_k$ for every $k$. This follows by reversing the above reasoning: for fixed $\sum_i a_i$ (i.e., fixed negativity), the quantity $m_k$ is maximized precisely when the Schmidt coefficients take the form  
$a_1 = \cdots = a_k = x$ and $a_{k+1} = \cdots = a_N = y$.  
Thus, the state maximizing $\mathcal{N}(\ketbra{\psi_{j^*}}{\psi_{j^*}})$ also maximizes $m_k$. The tightest possible inequality is obtained by selecting the decomposition of $\rho$ in which $\ket{\psi_{j^*}}$ simultaneously attains the smallest possible values of $\mathcal{N}$ and $m_k$ among all optimal decompositions.
\end{proof}
\begin{figure}[tbp]
    \centering
    \includegraphics[width=\linewidth]{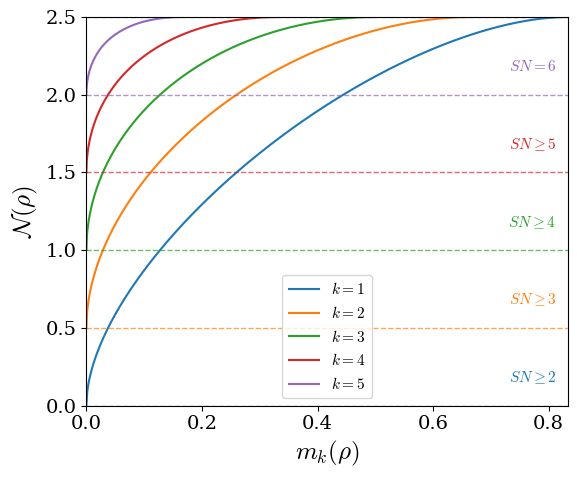}
    \caption{ Maximal negativity compatible with $m_k$, for mixed bipartite states of local dimension $N=6\leq M$ with respect to the measure $m_k$ in Eq.~\eqref{eq:PartialSums}. Here $m_k$ is defined according to Eq.~\eqref{eq:MixedMeasure}, and provides noise robustness to the criteria derived for Schmidt number. For each value of $k$, the horizontal dashed lines delimit the regions where the Schmidt number (SN) is above certain thresholds. 
    } 
    \label{fig:placeholder}
\end{figure}
For the case $M_{k'>k}=0$, the Schmidt rank of $\ket{\psi_{j^*}}$ is smaller or equal than $k$, and therefore $\rho$ has Schmidt number smaller or equal than $k$. This is a fundamentally distinguished case because no more than $k$ degrees of freedom need to be entangled in order to construct the state $\rho$, and therefore we shall devote the coming section to this case.

\section{Schmidt number from spectrum}
\label{sec:SNFS}
In this last section we address the characterization of the set of bipartite states whose SN cannot exceed a given value $\chi$, under the action of global unitaries. 
The SN is an entanglement measure in the form of Eq.~\eqref{eq:MixedMeasure}, which applies also to PPT-entangled states. Certification of SN can be done by using $\kappa$-positive but $(\kappa+1)$-negative maps \cite{terhal_schmidt_2000}, through 
SN-witnesses \cite{sanpera_schmidt-number_2001} (see Section~\ref{sec:SNW}), symmetric informationally complete measurements \cite{tavakoli_enhanced_2024, wang_schmidt_2024}, covariance matrix \cite{liu_bounding_2024} or hierarchy conditions of linear systems \cite{johnston_complete_2022}. Nevertheless, for full rank states, most of the known methods in the literature fail in bounding the entanglement properties of the states under scrutiny.

In accordance with the previous concepts, we introduce first the notion of Schmidt number from spectrum $\text{SNFS}_{\chi}$, which following Def.~\ref{def:EFS} reads (see also \cite{mallick_absolute_2026}):
\begin{definition}
\label{def:SNFS}
 A state $\rho \in \mathcal B(\mathbb{C}^{N}\otimes\mathbb{C}^{M})$ has $\chi$-SN from spectrum if, for any global unitary matrix $U\in \mathcal{U}(D)$, $SN(U\rho U^{\dag}) \leq \chi$ and the bound $\chi$ is saturated for at least one such unitary.
\end{definition}

Accordingly, $\mathrm{SNFS}_{\chi}$ denotes the set of states $\rho$ whose Schmidt number is at most $\chi$, certified directly from their spectra. The sets $\mathrm{SNFS}_{\chi}$ are convex, compact, and form a hierarchical nested structure $\mathrm{SNFS}_{1} \subset \mathrm{SNFS}_{2} \subset \cdots \subset \mathrm{SNFS}_{\min(N,M)}.
$
The innermost set, $\mathrm{SNFS}_{1}$, coincides with the separable-from-spectrum set (SEPFS), while the outermost $\mathrm{SNFS}_{\min(N,M)}$, spans the entire space $\mathcal{B}(\mathbb{C}^{N} \otimes \mathbb{C}^{M})$.

In what follows, we provide inner and outer characterizations of the sets $\mathrm{SNFS}_{\chi}$, all based on estimating the values of $\alpha_{+}(\chi)$ as close as possible to the tight threshold $\tilde{\alpha}_{+}(\chi)$, defined as the largest value of $\alpha_{+}(\chi)$ such that $\Lambda_{\alpha}(\rho) \in \mathrm{SNFS}_{\chi}$ for all $\rho$. To this end, we derive a lower bound for each set, along with two upper bounds, one based on the negativity and the other on the positive reduction criterion. These results are employed to certify which is the maximal SN of a state under global unitaries.

\subsection{$\text{SNFS}_{\chi}$ from Schmidt number robustness}
\label{sec:SNFSROB}
In analogy to the well known concept of robustness of a state \cite{vidal_robustness_1999}, it is possible to define a SN-robustness~\cite{clarisse_schmidt_2006} that quantifies how much a given state $\rho$ with $SN \leq\chi$  might be mixed while having a SN $\leq \chi$. Formally:  
\begin{definition}
\label{def:RandomSchmidtNumberRobustness}\cite{clarisse_schmidt_2006} The random SN-$\chi$ robustness is defined as the minimal $R_{\chi}$ such that
\begin{equation}
\label{eq:RandomSNRobustness}
\text{SN}\left(\frac{1}{1+R_{\chi}}(\rho+R_{\chi}\cdot\frac{\mathds{1}}{D})\right)\leq{\chi},
\end{equation}
for any possible state $\rho\in\mathbb{C}^{N}\otimes\mathbb{C}^{M}$.
\end{definition}

To set a lower bound we note that robustness is trivially related to the action of the reduction map introduced in Theorem~\ref{thm:InverseMapSNFS}. Then, if for a given family of states, namely $\rho_{F}$, we obtain a value $\alpha_+=\alpha_{F}$ such that $\Lambda_{\alpha_{F}}(\rho_{F})\in \text{SNFS}_{\chi}$ for $\alpha\in[-1,\alpha_{F}(\chi)]$, the value $\alpha_{F}(\chi)\leq \tilde{\alpha}(\chi)$ acts as a lower bound for the characterization of the $\text{SNFS}_{\chi}$ set, since the latter one has to fulfill Theorem~\ref{thm:InverseMapSNFS} for any $\rho\in\mathcal{B}(\mathbb{C}^{N}\otimes\mathbb{C}^{M})$, and not just for the restricted family $\rho_{F}$. 

In~\cite{clarisse_schmidt_2006}, one of such specific bounds for $\mathbb{C}^{N}\otimes\mathbb{C}^{N}$ was derived, which reads $\alpha_{LB}=\frac{2(N\chi-1)}{N-\chi}$. We use it in the following theorem:

\begin{theorem}
\label{thm:Alpha_LB_SNFS} 
 Let $\rho \in \mathcal B(\mathbb{C}^{N}\otimes\mathbb{C}^{N})$  be a normalized state of total dimension $D=N^2$, and $\boldsymbol{\lambda} = \{\lambda_i^{\uparrow} \}_{i=0}^{D-1}$ the corresponding eigenvalues in a non-decreasing order, with $\sum_{i=0}^{D-1}\lambda_i = 1$. If
\begin{equation}
\label{eq:Alpha_LB_SNFS}
\lambda_{0}\geq \frac{1}{D+\frac{2(N\chi-1)}{N-\chi}}, \quad\text{ or }
\end{equation}
 \begin{equation}
 \label{eq:Alpha_LB_SNFS2}
K \cdot \sum_{i=0}^{c-1} \lambda_i + \left[N^2 -K \cdot c + \frac{2(N\chi-1)}{N-\chi} \right] \cdot \lambda_c \geq 1,
\end{equation}
where $ K=  1 +\frac{2(N\chi-1)}{N-\chi} $, and $c = \left\lceil \frac{(N^2-1)(N-\chi)}{2(N\chi-1)+(N-\chi)} \right\rceil$, then $\rho\in\text{SNFS}_{\chi}$.\\
\end{theorem}
\begin{proof}
We make use of Lemma~\ref{Lemma:GeneralAlphaPM} considering $\alpha_{-}=-1$ and $\alpha_{LB}$, which relates with the robustness computed in~\cite{clarisse_schmidt_2006} through $\alpha_{LB}=1/R_{\chi}$.
\end{proof}
We note that Theorem~\ref{thm:Alpha_LB_SNFS} allows us to certify the SN of a state beyond SN robustness. Also in ~\cite{clarisse_schmidt_2006} a tight bound $\alpha_{C}(\chi)=2\cdot(2\chi^{2}-1)$ was conjectured. With this value, we extend the conjecture to generic states as follows:
\begin{conjecture}
\label{conj:Alpha_C_SNFS}
 Let $\rho \in \mathcal B(\mathbb{C}^{N}\otimes\mathbb{C}^{M})$  be a normalized state of total dimension $D=NM$ with $N\leq M$ and $\boldsymbol{\lambda} = \{\lambda_i^{\uparrow} \}_{i=0}^{D-1}$ the corresponding eigenvalues in a non-decreasing order, with $\sum_{i=0}^{D-1}\lambda_i = 1$. If
\begin{equation}
\label{eq:Alpha_C_SNFS}
\lambda_{0}\geq \frac{1}{D+2\cdot(2\chi^{2}-1)}, \quad\text{ or }
\end{equation} 
 \begin{equation}
 \label{eq:Alpha_C_SNFS2}
 \begin{split}
(4\chi^2-1) \cdot \sum_{i=0}^{c-1} \lambda_i &+ \Big [ D -(4\chi^2-1) \cdot c  \\ &+2(2\chi^2-1) \Big ] \cdot \lambda_c \geq 1,
\end{split}
\end{equation}
where $c = \left\lceil \frac{D-1}{4\chi^2-1} \right\rceil$, then $\rho\in\text{SNFS}_{\chi}$.\\
\end{conjecture}
Despite our efforts, we have not been able to complete the pending parts of the proof, as stated in Ref.~\cite{clarisse_schmidt_2006}, nor to find a counterexample to its claims with current state-of-the-art methods to certify SN \cite{liu_bounding_2024, johnston_complete_2022}. We have tried to confirm the SN of the states defined by $\alpha_{C}$ with up to $10^6$ unitary matrices distributed over the Haar measure, for different local dimensions $N$ and SN values $\chi$. Nevertheless, no contradiction has been found.

Finally, from~\cite{lewenstein_sufficient_2016} a general $\mathbb{C}^{N}\otimes\mathbb{C}^{M}$ lower bound of $\alpha_{+}=\chi+1$ has been derived, which we extend with our convex geometry approach.
\begin{theorem}
\label{thm:Alpha_L_SNFS} 
 Let $\rho \in \mathcal B(\mathbb{C}^{N}\otimes\mathbb{C}^{M})$  be a normalized state acting on a Hilbert space of global dimension $D=NM$ with $N\leq M$ and $\boldsymbol{\lambda} = \{\lambda_i^{\uparrow} \}_{i=0}^{D-1}$ the corresponding eigenvalues in a non-decreasing order, with $\sum_{i=0}^{D-1}\lambda_i = 1$. If
\begin{equation}
\label{eq:Alpha_L_SNFS}
\lambda_{0}\geq \frac{1}{D+\chi+1}, \quad\text{ or }
\end{equation}

 \begin{equation}
 \label{eq:Alpha_L_SNFS2}
(\chi+2) \cdot \sum_{i=0}^{c-1} \lambda_i + \left[D -(\chi+2) \cdot c + \chi+1 \right] \cdot \lambda_c \geq 1,
\end{equation}
where $c = \left\lceil \frac{D-1}{\chi+2} \right\rceil$, then $\rho\in\text{SNFS}_{\chi}$.\\
\end{theorem}
\begin{proof}
We make use of Lemma~\ref{Lemma:GeneralAlphaPM} considering $\alpha_{-}=-1$ and $\alpha_{+}=\chi+1$  derived in~\cite{lewenstein_sufficient_2016}.
\end{proof}
In general, the previous (lower) bound is not tight, however, it provides an inner characterization of the $\text{SNFS}_{\chi}$ sets. Moreover, such bound is significantly worse than the one derived from Theorem~\ref{thm:Alpha_LB_SNFS} when $N=M$. From our results on SEPFS, i.e., for $\chi=1$,  we recover the expected $\alpha_{LB / C}=2$. 
\subsection{$\text{SNFS}_{\chi}$ from negativity}
\label{sec:SNFSNEGFS}
Interestingly, the SN of any NPT mixed state can be upper bounded from NEG (Section~\ref{sec:NEGFS}).
\begin{lemma}
\label{Lemma:NegativitySN}
    For any state $\rho\in \mathcal B(\mathbb{C}^{N}\otimes\mathbb{C}^{M})$ of Schmidt number at most $\chi$, $\mathcal{N}(\rho)\leq (\chi-1)/2 $.
\end{lemma}
\begin{proof}
    By convexity of the negativity, it is sufficient to demonstrate the bound for pure states, addressed in Lemma~\ref{lemma:NegativityIsotropic}. 
\end{proof}
Lemma~\ref{Lemma:NegativitySN} also implies that, given a state $\rho\in\mathcal B(\mathbb{C}^{N}\otimes\mathbb{C}^{M})$, if $\mathcal{N}(\rho)>(\chi-1)/2$ then $\text{SN}(\rho)>\chi$, a condition reproduced in Fig~\ref{fig:Negativities} by the horizontal dashed lines. The lemma immediately leads to the following theorem characterizing this superset of $\text{SNFS}_{\chi}$:
\begin{theorem}
\label{thm:Alpha_NEG_SNFS} 
 Let $\rho \in \mathcal B(\mathbb{C}^{N}\otimes\mathbb{C}^{M})$  be a normalized state (with $N\leq M$ and $D=NM$)  and $\boldsymbol{\lambda} = \{\lambda_i^{\uparrow} \}_{i=0}^{D-1}$ the corresponding eigenvalues in a non-decreasing order, with $\sum_{i=0}^{D-1}\lambda_i = 1$ and $2\leq\chi\leq N-1$. If
\begin{equation}
\label{eq:Alpha_NEG_SNFS}
\lambda_{0}\geq \frac{1}{D+\frac{D(\chi-1)+N(N-1)}{N-\chi}}, \quad\text{ or }
\end{equation}
 \begin{equation}
 \label{eq:Alpha_NEG_SNFS2}
K \cdot \sum_{i=0}^{c-1} \lambda_i + \left[D -K \cdot c + \alpha_+ \right] \cdot \lambda_c \geq 1,
\end{equation}
where:  
\begin{equation}
K=\frac{\chi(D-1)+N(N-M)}{N-\chi},
\end{equation}
\begin{equation}
\alpha_{+}=\frac{D(\chi-1)+N(N-1)}{N-\chi},
\end{equation}
\begin{equation}
c = \left\lceil \frac{(D-1)(N-\chi)}{\chi(D-1)+N(N-M)} \right\rceil.
\end{equation}
Then, $\rho\in\text{NEGFS}_{(\chi-1)/2}$.\\
\end{theorem}
\begin{proof}
The proof is analogous to the one of Theorem~\ref{thm:AlphaNegativity} substituting $\gamma$ by $(\chi-1)/2$ as the SN that we want to certify. Notice that the crossing between the line corresponding to the vector of Schmidt coefficients $\mathbf{a}_{1}=(1/\sqrt{N},\cdots,1/{\sqrt{N}})$ and $\mathbf{a}_{2}=(1/\sqrt{N-1},\cdots,1/\sqrt{N-1},0)$ takes place at $\alpha=2(N-1)$. Nevertheless, the maximal negativity attainable under global unitaries for this value is $\frac{N^2-3N+2}{2\cdot N\cdot(M+2)-4}$, which is smaller than $0.5$ for any $N,M$. Thus, only considering the pure state with maximal Schmidt rank as the PPS is already enough. 
\end{proof}
We recall that the sets of $\text{NEGFS}_{(\chi-1)/2}$ are upper bounds of the $\text{SNFS}_{\chi}$ sets. The previous Theorem~\ref{thm:Alpha_NEG_SNFS} does not consider $\chi=1$, which recovers the PPTFS set (with $\alpha_{+}=2$ \cite{abellanet-vidal_sufficient_2025}), nor $\chi = N$, which recovers the whole set of states as $\alpha_{+}$ tends to infinite.\\

Similar conditions~\cite{eltschka_partial_2015} can be derived to bound the EM of concurrence $C(\rho)$ by its corresponding $SN(\rho)$:
\begin{equation}
\label{eq:ConcurrenceSN}
C(\rho) \leq \sum_j p_j C(\psi_j) \leq \sqrt{\frac{2\big(\mathrm{SN}(\rho) - 1\big)}{\mathrm{SN}(\rho)}},
\end{equation}
where $C(\rho)$ is obtained via the convex roof extension of the pure state measure $C(\ket{\psi})=\sqrt{2(1-\Tr(\rho_{A}^{2}))}$. For generic states, the concurrence is generally hard to evaluate, but there exist bounds based on the negativity (see also \cite{eltschka_negativity_2013, eltschka_partial_2015, verstraete_comparison_2001}):
\begin{equation}
\label{eq:concurrence_negativity}
2\sqrt{\frac{2}{\tilde\chi(\tilde\chi-1)}}\mathcal{N}(\rho)\leq \mathcal{C}(\rho)\leq 2\mathcal{N}^{\text{CREN}}(\rho).
\end{equation}
Here, $\tilde\chi$ is not the Schmidt number, but the maximal Schmidt rank over the pure states in the optimal decomposition of the state for the concurrence \cite{eltschka_negativity_2013, eltschka_partial_2015, verstraete_comparison_2001}.
Thus, our approach can also be used to bound the \textit{concurrence from the spectrum} of a state.
\subsection{$\text{SNFS}_{\chi}$ from positive reduction}
\label{sec:SNFSPREDFS}
\noindent Another useful tool to bound the SN is the reduction map: 
\begin{equation}
\label{eq:RED_General}
     \text{RED}_{\kappa}(\rho) = \Tr(\rho)\mathds{1}-\frac{1}{\kappa}\rho
\end{equation}
which is known to be  a $\kappa$-positive but $(\kappa+1)$-negative map, and thus can be used to certify $SN(\rho)$ \cite{terhal_schmidt_2000, horodecki_reduction_1999}:
\begin{theorem}\cite{terhal_schmidt_2000}
\label{thm:kpositivity}
Given a state $\rho\in\mathcal{B}(\mathbb{C}^{N}\otimes\mathbb{C}^{M})$, if
\begin{equation}
\label{eq:Reduction_CRITERION_SN}
\left[\mathds{1}\otimes\text{RED}_{\kappa}\right](\rho)\not\geq0,
\end{equation}
then $\text{SN}(\rho)>\kappa$.
\end{theorem}
Though Theorem~\ref{thm:kpositivity} does not provide  sufficient conditions to certify the SN from the spectrum, that is, whether $\rho\in \text{EFS}_{\gamma}$, it allows us to define a new superset of the desired $\text{SNFS}_{\chi}$, namely the set of states that have $\kappa=\chi$ positive reduction from spectrum, $\text{PREDFS}_{\kappa}$. To avoid confusion with the approach in Section~\ref{sec:HowTo}, we slightly change the notation for the reduction map.
\begin{definition}
\label{def:PREDFS}
 A state $\rho \in \mathcal B(\mathbb{C}^{N}\otimes\mathbb{C}^{M})$ has positive $\kappa$-reduction from spectrum ($\text{PREDFS}_{\kappa}$) if, for any global unitary matrix $U$ acting on $\mathbb{C}^{N}\otimes\mathbb{C}^{M}$, $\left[\mathds{1}\otimes\text{RED}_{\kappa}\right](U\rho U^{\dag}) \geq0 $.
\end{definition}
Remarkably $\text{SNFS}_{\chi}\subset\text{PREDFS}_{\chi}$, since Theorem~\ref{thm:kpositivity} is sufficient, but not necessary for SN detection. 
 
 We approach these sets as before, namely, considering the range of $\alpha$ ( as in Theorem~\ref{thm:InverseMapSNFS}) such that $\Lambda_{\alpha}(\rho)\in\text{PREDFS}_{\chi}$. In this way, we 
 provide an upper bound for the $\text{SNFS}_{\chi}$ characterization, i.e., $\tilde\alpha\leq\alpha_{RED}$. 
 
 Positive reduction from spectra has been previously considered for $\kappa=1$ \cite{jivulescu_positive_2015}, and it is based on the positivity of any unitarily transformed state under the action of a positive but not completely positive map, introduced first in~\cite{hildebrand_PPT_2007} for the transposition map. Moving beyond the specific certification of $\text{PREDFS}_{1}$ addressed in \cite{jivulescu_positive_2015}, we develop a generalized framework for $\text{PREDFS}_{\kappa}$ valid for any $\kappa$. This transition from the unit case to the arbitrary $\kappa$ regime introduces significant technical complexities. We detail the specific conditions below, with the full analytical proofs and derivation strategies provided in Appendix~\ref{Appendix:DerivationPREDFS}.

\begin{theorem}
\label{thm:Our4.2.}
Let  $\rho \in \mathcal B(\mathbb{C}^{N}\otimes\mathbb{C}^{M})$  be a normalized state in a Hilbert space of global dimension $D=NM$. Let $\{\lambda_{i}^{\uparrow}\}$ be the eigenvalues of $\rho$ in non-decreasing order and $\{\gamma_{j}^{\downarrow}\}$ the eigenvalues of $\left[\mathds{1}\otimes\text{RED}_{\kappa}\right](\ketbra{\psi}{\psi})$ in non-increasing order (as introduced in Theorem~\ref{thm:Our3.1.} of Appendix~\ref{Appendix:DerivationPREDFS}). If for any possible pure state $\ket{\psi}$
\begin{equation}
\label{eq:Our4.2.}
\sum_{i=1}^{D}\lambda_{i}^{\uparrow}\cdot\gamma_{i}^{\downarrow}\geq 0,
\end{equation}
then $\rho\in\text{PREDFS}_{\kappa}$. 
\end{theorem}
\begin{proof}
(See also~\cite{jivulescu_positive_2015},\cite{hildebrand_PPT_2007}). We delegate some technical details of the extension to general values of $k$ to Appendix~\ref{Appendix:DerivationPREDFS}. 

\noindent Denoting $\left[\mathds{1}\otimes\text{RED}_{\kappa}\right]:=\xi_\kappa$,  the following chain of equations are equivalent:
\begin{equation}
\label{eq:4.2.1}
    \xi_{\kappa}(U\rho U^{\dag})\geq 0, \forall U,
\end{equation}
\begin{equation}
\label{eq:4.2.2}
    \Tr[\xi_{\kappa}(U\rho U^{\dag})\ketbra{\psi}{\psi}]\geq 0, \quad \forall U, \forall \ket{\psi}
\end{equation}
\begin{equation}
\label{eq:4.2.3}
\min_{U} \Tr[U\rho U^{\dag}\xi_{\kappa}^{\dag}(\ketbra{\psi}{\psi})]\geq 0, \quad \forall\ket{\psi}
\end{equation}
\begin{equation}
\label{eq:4.2.4}
    \min_{U} \langle U\rho U^{\dag},\xi_{\kappa}(\ketbra{\psi}{\psi})\rangle = \sum_{i=1}^{D}\lambda_{i}^{\uparrow}\gamma_{i}^{\downarrow}\geq 0
\end{equation}
For the last implication, we have made use of Observation~\ref{obs:selfadjoint} of the Appendix~\ref{Appendix:DerivationPREDFS}.
\end{proof}

Notice that Theorem~\ref{thm:Our4.2.} needs to be fulfilled $\forall \ket{\psi}$, i.e., for any possible vector of Schmidt coefficients $\{a_{i}\}$. Thus, even though we propose a complete characterization of the sets, it is needed to compute an infinite amount of inequalities in order to certify a state as $\text{PREDFS}_{\kappa}$. To compute the $\alpha_{RED}$ such that it encloses the set of $\text{PREDFS}_{\kappa}$, we consider the following proposition on isotropic states.
\begin{corollary}
\label{coro:PREDFSCHI} 
 Let $\rho \in \mathcal B(\mathbb{C}^{N}\otimes\mathbb{C}^{M})$  be a normalized state acting on a Hilbert space of global dimension $D=NM$ with $N\leq M$ and $\boldsymbol{\lambda} = \{\lambda_i^{\uparrow} \}_{i=0}^{D-1}$ the corresponding eigenvalues in a non-decreasing order, with $\sum_{i=0}^{D-1}\lambda_i = 1$. If
\begin{equation}
\label{eq:PREDFSCHI}
\lambda_{0}\geq \frac{1}{D+\frac{N\cdot(M\kappa-1)}{N-\kappa}}, \quad\text{ or }
\end{equation}
 \begin{equation}
 \label{eq:PREDFSCHI2}
 \begin{split}
\frac{\kappa(D-1)}{N-\kappa} \cdot \sum_{i=0}^{c-1} \lambda_i + \\+\Big [D -\frac{\kappa(D-1)}{N-\kappa} \cdot c  + \frac{N(M\kappa-1)}{(N-\kappa)} \Big ] \cdot \lambda_c \geq 1,
\end{split}
\end{equation}
where  $c = \left\lceil \frac{N-\kappa}{\kappa} \right\rceil$, then $\rho\in\text{PREDFS}_{\kappa}$
\end{corollary}
\begin{proof}
The derivation of the parameter $\alpha_{RED}(\kappa)$ for isotropic states is detailed in Proposition~\ref{Prop:Our6.1.} of the Appendix~\ref{Appendix:DerivationPREDFS} following the structure of \cite{jivulescu_positive_2015}.
\end{proof}
\noindent As expected, we recover the results of \cite{jivulescu_positive_2015} when we consider $\kappa=1$.

Finally, we find states for which the SN can be completely determined using only the reduction map. Considering the sufficient condition in Theorem~\ref{thm:kpositivity}, we find $\text{SN}(\rho)>\chi-1$. Moreover, using our Theorem~\ref{thm:Alpha_LB_SNFS}, we upper bound it as $\text{SN}(\rho)\leq\chi$. Thus, we certify $\text{SN}(\rho)=\chi$ with the use of the single mathematical tool of the reduction map with two different approaches. In specific, it is possible to find spectra with $\lambda_{\min}=1/(N^2+\alpha_{LB}(\chi))$ such that the maximal negativity of the reduction map with $\kappa=\chi-1$ under unitaries (see Eq.\eqref{eq:NegativityReduction} of Appendix~\ref{sec:NegativityReduction}) is greater than $0$. 

\subsection{Comparison of the $\text{SNFS}_{\chi}$ bounds}
At this stage, it is useful to summarize and compare our results obtained so far. On one hand, Lemma~\ref{Lemma:NegativitySN} together with Theorem~\ref{thm:Alpha_NEG_SNFS} and the results from Theorem~\ref{thm:Our4.2.}, can  be used to test the validity of the conjectured values of $\alpha_{C}(\chi)$ in Conjecture~\ref{conj:Alpha_C_SNFS}. In this regard, we find that the supersets of $\text{SNFS}_{\chi}$ always have larger values of the parameter $\alpha$, which motivates the search for tighter bounds, breaking or validating Conjecture~\ref{conj:Alpha_C_SNFS}. In the following, we collect the different bounds we obtained for $\alpha$ for the $\text{SNFS}_{\chi}$ sets and assess how they compare among them and with $\alpha_C$.   
\begin{table}[h]
    \centering
    \begin{tabular}{|c|c|c|}
    \hline
       Technique  &  $\alpha_{+}(\chi)$ \\
       \hline
        Robustness (LB) &$\frac{2(N\chi-1)}{N-\chi}, \chi+1 $ \\
            Conjectured (C) & $2\cdot(2\chi^2-1)$ \\
      Negativity (NEG) & $\frac{D(\chi-1)+N(N-1)}{N-\chi}$ \\
          Positive reduction (RED) & $\frac{N(M\chi-1)}{N-\chi}$  \\
        \hline
    \end{tabular}
    \caption{Bounds on $\alpha$ for $\mathrm{SNFS}_{\chi}$ derived from different techniques.  
    The bound from SN robustness (Section~\ref{sec:SNFSROB}) provides a lower bound for $\tilde\alpha$. Negativity (Section~\ref{sec:SNFSNEGFS}) and positive reduction (Section~\ref{sec:SNFSPREDFS}) provide sufficient conditions to certify $\text{SN}>\chi$, thus providing upper bounds on $\tilde\alpha$. We remind the reader that the robustness bound $\frac{2(N\chi-1)}{N-\chi}$ is stated for $\mathbb{C}^N\otimes\mathbb{C}^N$ states only and that the cases with $\chi=1$ and $\chi=N$ are not considered.}
    \label{tab:summary}
\end{table}
Despite we cannot give a complete characterization of the sets, Table~\ref{tab:summary} provides a comparison between the different sets for the pseudo-pure states, PPS, ranging from pure states ($p=0$) to the MMS ($p=1$). Explicitly,
\begin{equation}
\label{eq:comparison_sets}
\alpha_{LB}\leq\alpha_{C}\leq\alpha_{NEG}\leq\alpha_{RED}.
\end{equation}
Moreover, the complete inclusions of the considered $\text{EFS}_{\gamma}$ sets is only pending for the unknown case $\text{NEGFS}_{(\chi-1)/2}\overset{?}{\subset} \text{PREDFS}_{\chi}$, which is also unknown for the non-FS case. It is also interesting to note that the $\text{NEGFS}_{(\chi-1)/2}$ set does not come from a $\kappa$-positive map as the $\text{PREDFS}_{\kappa}$, but from simpler convexity relations using a $1$-positive map and the Schmidt decomposition of different pure states. On the other hand, since $\alpha_{NEG}\leq\alpha_{RED}$,  to detect $\mathrm{SNFS}_{\chi}$ it is always better to use the $\mathrm{NEGFS}$ bound. Despite these limitations, the study of the $\text{PREDFS}_{\chi}$ is still interesting on its own. 
\subsection{Spectral properties of Schmidt number witnesses}
\label{sec:SNW}

We finish our study by exploiting the well known fact that the characterization of SN has a dual description in terms of Schmidt number witnesses SNWs \cite{sanpera_schmidt-number_2001}.  A SNW corresponds to a hyperplane separating the convex set of states with a certain SN $\chi$ and an exterior point. For $\chi=1$, one recovers the well known entanglement witnesses~\cite{horodecki_separability_1996,terhal_family_2001,lewenstein_characterization_2001}.
They are defined as follows:
\begin{definition}
Let $\text{SN}_{\chi}$ denote the set of bipartite quantum states with Schmidt number at most $\chi$.
A Schmidt number witness for Schmidt number $\chi$ is a Hermitian operator $W_\chi$ such that $\Tr(W_\chi \sigma) \geq 0$, $\forall \sigma \in \text{SN}_{\chi}$.
Moreover, there exists at least one state $\rho$ with $\text{SN}(\rho)>\chi$ such that $\Tr(W_\chi \rho) < 0$.
\end{definition}

In this way, a $\chi$-Schmidt number witness $W_{\chi}$ detects states whose SN is strictly larger than $\chi$ by separating them from the convex set $\text{SN}_{\chi}$. A complete characterization of the $\text{SNFS}_{\chi}$ sets can be translated into constraints on the spectrum of admissible witnesses. Indeed, for a state $\rho$ with spectrum $\{\lambda_{i}^{\uparrow}\}$, it follows that 
\begin{equation}
\min_U \Tr\left(U\rho U^\dag W_{\chi}\right)=
\sum_i \lambda_i^{\uparrow}\mu_i^{\downarrow}
\geq 0,
\end{equation}
for every $\chi$-SN witness $W_{\chi}$, where $\{\mu_{i}^{\downarrow}\}$ denotes the spectrum of $W$.  To the best of our knowledge, this dual spectral viewpoint has so far been investigated only for EW and SEPFS in Ref.~\cite{johnston_inverse_2018}.

Extensive work has been devoted to determine important properties of the witnesses (decomposability, optimality) from their eigenvalues~\cite{sarbicki_spectral_2008,chruscinski_spectral_2009,champagne_spectral_2022,johnston_inverse_2018, song_spectral_2026,ende2025CritOptWits}. Here, we provide further criteria to find bounds on the eigenvalues of Schmidt number witnesses.  In particular, from Theorem~\ref{thm:InverseMapSNFS}, it 
follows:
\begin{theorem}
\label{thm:WitnessMIN}
Let $W_{\chi}$ be a $\text{SN}_{\chi}$ witness of states acting in $\mathbb{C}^{N}\otimes\mathbb{C}^{M}$. Then, it fulfills:
\begin{equation}
\label{Eq:413}
    W_{\chi} + \frac{\Tr(W_{\chi})}{\alpha_{+}} \mathds{1} \geq 0, 
\end{equation}
i.e., its minimal eigenvalue is at least $\lambda_{\min}(W_{\chi})\geq-\Tr(W_{\chi})/\alpha_{+}$, where $\alpha_{+}$ is the SN bound or a lower bound for it (see Table~\ref{tab:summary}).
\end{theorem}

\begin{proof}
Let $\Delta$ be a generic state, i.e., a unit trace PSD matrix acting in $\mathbb{C}^{N} \otimes \mathbb{C}^{M}$. Then,
\begin{equation}
    \sigma = \left(1 - \frac{D}{D+ \alpha_{+}}\right) \Delta + \frac{\mathds{1}}{D+\alpha_{+}}, \end{equation}
will be detected as $\text{SN}_{\chi}$ according to Theorem~\ref{thm:InverseMapSNFS}. Consequently, for any $\text{SN}=\chi$ witness, $\Tr(W_{\chi} \sigma) \geq 0$. We derive the condition
\begin{equation}
    \Tr \left[\left(W_{\chi} + \frac{\Tr(W_{\chi})}{\alpha_{+}} \mathds{1} \right)\Delta\right]  \geq 0.
\end{equation}
As the previous inequality is valid for any positive semidefinite $\Delta$, it implies the claimed result. 
\end{proof}

For the specific case of detecting entanglement from separability, the bound $\lambda_{\min}(W_{1})\geq -\Tr(W_{1})/2$ was only known to hold for decomposable entanglement witnesses \cite{johnston_inverse_2018,rana_negative_2013}, those that can be expressed as $\text{EW}=P+Q^{T_{A}}$, with $P,Q\geq 0$ \cite{horodecki_bound_1999, terhal_detecting_2002}. Consequently, given the separability threshold at $\alpha_{+}=2$, we extend the previous result to non-decomposable entanglement witnesses \cite{song_spectral_2026}. Moreover, this approach is valid to bound the entanglement of $W_{\chi}$ witnesses for arbitrary $\chi$.

Finally, by considering the lower bound of $\alpha_{-}$, it is also possible to upper bound the biggest eigenvalue of any witness.
\begin{theorem}
\label{thm:WitnessMAX} Let $W_{\chi}$ be a $\text{SN}=\chi$ witness in $\mathbb{C}^{N}\otimes\mathbb{C}^{M}$ with maximal eigenvalue $\lambda_{\max}$. Then, the following relation is fulfilled:
\begin{equation}
\label{eq:WMAX}
\lambda_{\max}(W_{\chi})\leq\Tr(W_{\chi}).
\end{equation}
\end{theorem}
\begin{proof}
The proof is analogous to that of Theorem~\ref{thm:WitnessMIN}, but considering $\alpha=-1$ and the corresponding flip in the inequalities instead.
\end{proof}

\section{Conclusions}
\label{conclusions}
In this work, we introduced a new paradigm in entanglement certification from the notion of $\gamma$-entanglement from spectrum, $\mathrm{EFS}_{\gamma}$, defined as the sets of bipartite mixed states whose entanglement content, with respect to a given entanglement measure EM, cannot exceed $\gamma$ under any global unitary transformation. As a consequence, membership in $\mathrm{EFS}_{\gamma}$ is fully determined by the eigenvalue spectrum of the state. Building on our previous techniques based on linear maps and their inverses, we demonstrated that these tools also yield powerful characterizations of such entanglement-constrained sets. Importantly, our criteria require only a small number of eigenvalues, making them applicable even in scenarios where the available information about the state is very limited. 

We focused on two widely used bipartite entanglement measures, the negativity and the Schmidt number, and derived bounds characterizing the corresponding families of sets $\mathrm{EFS}_{\gamma}$. Our methods, however, are general and can be applied to other entanglement measures as well. Our analysis further shows that the negativity is particularly well suited for the spectral characterization of entanglement. We also remark that obtaining tight bounds for Schmidt number sets, $\mathrm{SNFS}_{\chi}$, remains an open problem, even for pseudo-pure states (PPS). Additionally, we examined the relationship between the two measures and studied the set of states with positive reduction from spectrum as an upper bound for those with bounded Schmidt number under global unitaries.

Our criteria are specially tailored to full-rank and highly mixed entangled states, which are among the most challenging to characterize using state-of-the-art methods. We remark that such states are experimentally relevant, as they correspond to pure states affected by depolarizing noise, a scenario of significant practical importance. Further investigation of low-rank states as well as alternative noise models is needed to gain a deeper understanding of the effect of quantum channels when one only has a limited knowledge of the state. Finally, extending our techniques to multipartite systems offers a promising direction for future research. \\

\noindent \textbf{Acknowledgments.--}
We want to acknowledge fruitful conversations with B. Mallick, S. Mukherjee, N. Ganguly and A.S. Majumdar. We thank K.N.B. Teja, J. Ahiable, G. Rajchel-Mieldzioć and K. Życzkowski for valuable discussions. \\

JAV acknowledges financial support from Ministerio de Ciencia e Innovación of the Spanish Government FPU23/02761. GMR acknowledges financial support by the European Union
under ERC Advanced Grant TAtypic, Project No. 101142236. JAV, AR and AS acknowledge financial support from Spanish MICIN (projects: PID2022:141283NBI00) with the support of FEDER funds, the Spanish Government with funding
from European Union NextGenerationEU (PRTR-C17.I1), the Generalitat de Catalunya,
the Ministry for Digital Transformation and of Civil Service of the Spanish Government through the QUANTUM ENIA project -Quantum Spain Project- through the Recovery, Transformation and Resilience Plan NextGeneration EU within the framework
of the Digital Spain 2026 Agenda. AR acknowledges support from the Deutsche Forschungsgemeinschaft (DFG, German Research Foundation, project number 563437167), the Sino-German Center for Research Promotion (Project M-0294), and the German Federal Ministry of Research, Technology and Space (Project BeRyQC, Grant No.\ 13N17292). \\

\noindent \textbf{Author contribution statements.--}
All authors contributed to the development of the main ideas and theoretical framework. JAV performed the main calculations and wrote the first draft of the manuscript. GMR and AR contributed sections of the manuscript and were actively involved throughout the research and writing process. AS supervised the project, assisted with the writing process, and guided the overall research direction. All authors contributed to the work, read and discussed the manuscript, and approved the final version. Large language model tools were used to assist with typo correction, grammar checking, coding, and the development of figures.

\bibliographystyle{quantum}
\bibliography{TFM}

\appendix
\section{Negativity of the reduction map from the spectrum}
\label{sec:NegativityReduction}
In this section we extend the calculations of Section~\ref{sec:NEGFS} to the use of other positive but not completely positive maps. 
In specific, we address the sum of the negative eigenvalues of the action of the reduction map 
\begin{equation}
\xi_{\kappa}(\rho)=\left[\mathds{1}\otimes\text{RED}_{\kappa}\right](\rho),
\end{equation}
which combines the notion of negativity addressed in Section~\ref{sec:NEGFS} and the reduction map which we employ to bound the $\text{SNFS}_{\chi}$ set in Section~\ref{sec:SNFS}.\\ 

\begin{definition}
\label{def:NegativityReduction}
We define the 
reduction map negativity as
\begin{equation}
\label{def:NegativityReductionMap}
\mathcal{N}_{red}=\frac{||\xi_{\kappa(\rho)}||_{1}-(M-1/\kappa)}{2}.
\end{equation}
\end{definition}
This function is convex, due to convexity of the trace norm and linearity of the trace, and can be applied to any other positive but not completely positive map easily. We consider its maximal bounds over unitary transformations introducing the negativity of the reduction map from spectrum ($\text{NRFS)}$. \\

\begin{figure}[h!]
    \centering
    \includegraphics[width=\linewidth]{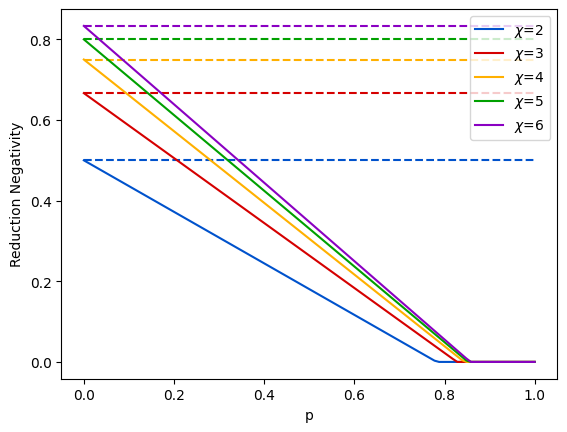}
    \caption{Maximal reduction map negativity for the map parameter $\kappa=1$ of PPS of different Schmidt Rank $\chi$ as a function of $p =D/(D+\alpha)$ for $N = M = 6$. }
    \label{fig:ReductionMapNegativity}
\end{figure}

As we have shown in Theorem~\ref{thm:Our3.1.}, $\xi_{\kappa}(\ketbra{\psi}{\psi})$ has at most $1$ negative eigenvalue. Thus, the same calculations from the spectrum can equivalently be done by considering $\mathcal{N}_{red}^{'}=\min\{0,\lambda_{\min}(\xi_{\kappa}(\rho))\}$, as the same maximal negativity of the reduction map can be achieved if we assume that all the negativity is concentrated in the smallest possible eigenvalue. 
Moreover, this single negative eigenvalue fulfills $\eta_{q}\geq -\frac{1}{\chi}+\frac{1}{\kappa}$ (it is negative if $\chi>\kappa$, i.e., we want to certify with a bigger $SR$ $\chi$ than the map parameter $\kappa$). The reduction map is linear, thus $\mathds{1}\otimes\text{RED}_{\kappa}$ will also be linear, yielding the following Lemma.
\begin{lemma}
\label{lemma:NegativityIsotropicReduction}
Given a normalized pure state $\ket{\psi}=\sum_{i}a_{i}\ket{ii}$ in $\mathbb{C}^{N}\otimes\mathbb{C}^{M}$, the reduction map negativity of $\rho=\frac{1}{D+\alpha}\left(\mathds{1}+\alpha\ketbra{\psi}{\psi}\right)$
is given by
\begin{equation}
\label{eq:NegativityReduction}
\mathcal{N}_{red}(\rho)=-\min\left[0,\frac{1}{D+\alpha}\left(\bigg (M-\frac{1}{\kappa}\bigg )+\alpha\eta_{q}\right)\right]
\end{equation}
where $\eta_{q}$ is the smallest eigenvalue of $[\mathds{1}\otimes\text{RED}_{\kappa}](\ketbra{\psi}{\psi}).$ 
Moreover, if $\ket{\psi}$ has SR $\chi$, then $\mathcal{N}_{red}(\rho)$ is saturated by $\mathbf{a}=(1/\sqrt{\chi},\cdots,1/\sqrt\chi,0,\cdots,0)$, which attains the extreme value of $\eta_{q}$.
\end{lemma}

\begin{proof}
It is first needed to consider the action of the linear map onto the MMS, as $[\mathds{1}\otimes\text{RED}_{\kappa}](\mathds{1})=(M-1/\kappa)\mathds{1}$. The action of the map onto a general pure state is described in detail in the Appendix~\ref{Appendix:DerivationPREDFS}.
\end{proof}

From the previous result, it is possible to state the following theorem regarding the negativity of the reduction map.
\begin{theorem}
\label{thm:AlphaNegativityReduction} 
Let $\rho \in \mathcal B(\mathbb{C}^{N}\otimes\mathbb{C}^{M})$  be a normalized state acting on a Hilbert space of global dimension $D=N M$ with $N\leq M$ and $\boldsymbol{\lambda} = \{\lambda_i^{\uparrow} \}_{i=0}^{D-1}$ the corresponding eigenvalues in a non-decreasing order, with $\sum_{i=0}^{D-1}\lambda_i = 1$. 
If
\begin{equation}
\lambda_{0}\geq \frac{1}{D+\frac{\chi(\kappa(M-\gamma D)-1)}{\gamma\chi\kappa+\chi-\kappa}}, \quad\text{or}
\end{equation}
 \begin{equation}
 \label{eq:GeneralAlphaPMReduction}
 \begin{split}
 \frac{\kappa(\chi(\gamma+M-\gamma D)-1)}{\gamma\chi\kappa+\chi-\kappa} \cdot \sum_{i=0}^{c-1} \lambda_i +\\
 + \Big [D - \frac{\kappa(\chi(\gamma+M-\gamma D)-1)}{\gamma\chi\kappa+\chi-\kappa}\cdot c + \\ 
+\frac{\chi(\kappa(M-\gamma D)-1)}{\gamma\chi\kappa+\chi-\kappa} \Big ] \cdot \lambda_c \geq 1,
\end{split}
\end{equation}
then $\rho\in \text{NRFS}_{\gamma}$, where $c = \left\lceil\frac{(D-1)(\gamma\kappa\chi+\chi-\kappa)}{\kappa(\gamma\chi+\chi M -\gamma\chi D -1)}\right\rceil$.
\end{theorem}

\begin{proof}
Given a pure state of Schmidt rank $\chi$, $\ket{\psi_\chi}$, the image of the reduction map is $\Lambda_\alpha(\ketbra{\psi_\chi}{\psi_\chi}) = \alpha\ketbra{\psi_\chi}{\psi_\chi} + \mathbb{I}:=\rho_\chi$, which corresponds to the unnormalized PPS addressed in Lemma~\ref{lemma:NegativityIsotropicReduction}. On the other hand, $\alpha_{-} = -1$ is considered, as discussed in Section~\ref{sec:HowTo}. 
\end{proof}

\begin{figure}[htpb!]
    \centering
    \includegraphics[width=\linewidth]{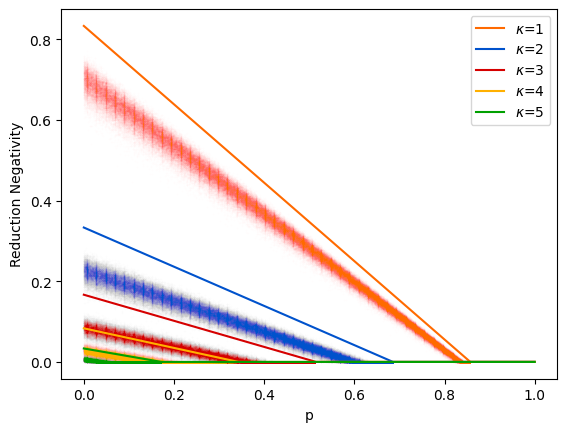}
    \caption{Maximal reduction map negativity given different map parameter $\kappa$ of PPS of maximal Schmidt rank $\chi=6$ as a function of $p =D/(D+\alpha)$ for $N = M = 6$. We also depict the reduction map negativities obtained for each spectrum under the action of $10^4$ Haar random unitary matrices.}
    \label{fig:ReductionsNegativities}
\end{figure}
In Fig.~\ref{fig:ReductionsNegativities} one observes a significant difference between the average negativity generated by Haar-random unitaries and the analytical bounds, as commented also in Section~\ref{sec:NEGFS}. Lemma~\ref{lemma:NegativityIsotropicReduction} provides the states that attain the desired maximum.

\section{Pseudo pure states}
\label{sec:PPS}
It is specially interesting that the values of $\alpha_{+}$ that we compute throughout the work are only tight for PPS. Therefore, by knowing the noise parameter of the pure state, it is possible to deduce the maximal entanglement that can be generated under global unitaries.
In the following, we detail how such a parameter can be inferred in practical scenarios. 
We consider the case of pure states $\ketbra{\psi}{\psi}$ affected by the depolarizing channel with certain probability, i.e., a PPS
\begin{equation}
    \rho=(1-p)\ketbra{\psi}{\psi}+p\frac{\one}{D}\,.
\end{equation}
Notice that $p=\frac{D}{D+\alpha_{+}}$. 
This case is particularly tractable analytically (see e.g. Fig.~\ref{fig:Negativities}) and can model any type of noise in a worst-case scenario, assuming that all gates spanning an operator basis can occur with equal probability. Then the spectrum depends on a single parameter.
If the target state $\ket{\psi}$ is known, one can compute the fidelity
\begin{equation}
    F=\bra{\psi}\rho\ket{\psi} = (1-p) + \frac{p}{D},
\end{equation}
and obtain the noise parameter $p$ with error $\varepsilon$ by using $n\propto O(\varepsilon^{-2})$ measurements with identical copies, since the variance scales with the root of the number of measurements $1/\sqrt{n}$. 
Alternatively, assuming that one has access to a device able to prepare two simultaneous copies of the state at each shot, one can use the swap test through the purity $\Tr(\rho^2)=\Tr(\rho^{\otimes 2}V)$ where $V$ is the SWAP operator. The purity can also be estimated with sequential copies with randomized measurements~\cite{elben_mixed-state_2020,neven_symmetry-resolved_2021}. In either case we have $\Tr(\rho^2)=(1-p)^2+2p(1-p)/D+p^2/D$. Similarly as in fidelity estimation, the error decreases as $O(\varepsilon^{-2})$ with the number of measurements.\\

Notice that this strategy does not apply to the case where we have a source of randomly unknown quantum states, distributed over the Haar measure, with fixed depolarizing noise parameter $p$. Given any pair of such states, $\rho=(1-p)\ketbra{\psi}+p\one/D$ and $\sigma=(1-p)\ketbra{\phi}+p\one/D$, we have that
\begin{align}
    \mathbb{E}(\Tr(\rho\sigma)) &= (1-p)^2\mathbb{E}(|\bra{\psi}\phi\rangle|^2)+\frac{p(2-p)}{D}\nonumber\\
    &= \frac{1}{D}\,,
\end{align}
since the average overlap between Haar random states is $\mathbb{E}(|\bra{\psi}\phi\rangle|^2)=1/D$. Thus one cannot gain information about $p$ by directly computing the expectation value with respect to the prepared states.

\section{Derivation of the PREDFS conditions}
\label{Appendix:DerivationPREDFS}
Here we derive in detail the conditions for $\text{PREDFS}_{\kappa}$ presented in Section~\ref{sec:SNFSPREDFS}. The proofs are based on Ref.~\cite{hildebrand_PPT_2007} and \cite{jivulescu_positive_2015}, where the authors derive similar bounds for the specific case of separability, namely for $\kappa=1$. We present here a generalization to arbitrary $\kappa$ of their calculations.\\

First, we recall that the map $\left[\mathds{1}\otimes\text{RED}_{\kappa}\right](\rho)=[\mathds{1}\otimes\text{RED}_{\kappa}](\rho)$ is self-adjoint. 
\begin{observation}
\label{obs:selfadjoint}
The reduction map applied on a subsystem of a state in a $\mathbb{C}^{M}\otimes\mathbb{C}^{N}$ Hilbert space is self-adjoint, namely for all $X,Y\in\mathcal{B}(\mathbb{C}^{M}),\mathcal{B}(\mathbb{C}^{N})$ it satisfies that
\begin{equation}
\label{eq:SelfAdjoint}
 \Tr\left(X^{\dag}\left[\mathds{1}\otimes\text{RED}_{\kappa}\right](Y) \right) = \Tr\left(\left[\mathds{1}\otimes\text{RED}_{\kappa}\right](X)^{\dag}Y \right) 
\end{equation}
\end{observation}
\begin{proof}
It is possible to consider the case of simple tensors $X=X_{A}\otimes X_{B}$ and $Y=Y_{A}\otimes Y_{B}$ due to antilinearity in $X$ and linearity in $Y$ of both expressions. The conclusion follows from direct computation.
\end{proof}
Secondly, we would like to calculate the spectra of the action of the reduction map on a subsystem of a general pure state, which was only known for $\kappa=1$.

\begin{theorem}
\label{thm:Our3.1.}
Let $\ket{\psi}\in\mathbb{C}^{N}\otimes\mathbb{C}^{M}$ be a normalized pure state with Schmidt decomposition, up to local unitaries,
\begin{equation}
\label{eq:SchmidtDecompositionPREDFS}
\ket{\psi}=\sum_{i=1}^{r}a_{i}\ket{ii}
\end{equation}
where $\{a_{i}>0\}$ and $1\leq r\leq \min(N,M)$ is the SR of $\ket\psi$. Suppose that the sets $\{a_{i}\}_{i=1}^{r}$ and $\{b_{1}>b_{2}>\cdots>b_{q}\}$ (where each $b_{i}$ has multiplicity $m_{i}$, $i=1,\cdots,q$) coincide. Then, the eigenvalues of the $\kappa$-reduced projection on $\left[\mathds{1}\otimes\text{RED}_{\kappa}\right](\ketbra{\psi}{\psi})$ fulfill the following list of conditions:
\begin{enumerate}
    \item Each $a_{i}^2$ appears with multiplicity $m_{i}\cdot M -1$.
    \item There are $q$ simple eigenvalues $\eta_1 > \cdots > \eta_q$, defined as the $q$ real solutions of $F_{x}(y)=0$, where
\begin{equation}
\label{eq:FXY}
    F_x(y) := 1 - \frac{1}{\kappa}\sum_{i=1}^{q}\frac{m_i b_i^2}{b_i^2 - y}.
\end{equation}
    \item The values interlace according to
\begin{equation}
\label{eq:OrderedEigenvaluesReduction}
    b_1^2 > \eta_1 > b_2^2 > \eta_2 > \cdots > b_q^2 > \eta_q,
\end{equation}
    where the smallest one satisfies
\begin{equation}
\label{eq:EtaqExpression}
\eta_q = 1 - \frac{1}{\kappa} - \sum_{i=1}^{q-1}\eta_i - \sum_{i=1}^q (m_i-1)\cdot b_i^2.
\end{equation}
    \item The null eigenvalue has multiplicity $(N-r)\cdot M$.  
    \item The auxiliary block
\begin{equation}
\label{eq:AuxiliaryMatrix}
A_{ij}= \big(a_i^2\delta_{ij} - \tfrac{1}{\kappa}a_i a_j\big)_{i,j\in(1,..., r)}
\end{equation}
has at most one negative eigenvalue. If $r\leq \kappa$, then $A\geq 0$. If $r>\kappa$, its unique negative eigenvalue $\eta_{q}$ satisfies
\begin{equation}
\label{eq:BoundEtaQ}
\eta_{\min}
\geq
-\left(\frac{1}{\kappa}-\frac{1}{r}\right),
\end{equation}
with equality if and only if the nonzero Schmidt coefficients are uniform,
$a_{i}=1/\sqrt{r}$.
\end{enumerate}
Finally, if $r>\kappa$, then necessarily $\eta_{q}<0$.
\end{theorem}

\begin{proof}
The argument follows closely Theorem $3.1.$ and Lemma $3.2.$ of \cite{jivulescu_positive_2015}. We will comment the main differences when extending to arbitrary natural values of $\kappa$, since the action of $\left[\mathds{1}\otimes\text{RED}_{\kappa}\right](\ketbra{\psi}{\psi}) = \sum_{i=1}^r a_{i}^{2} \ketbra{e_{i}}{e_{i}} \otimes \mathds{1}_M - 1/\kappa \cdot \sum_{i,j}a_ia_j\ketbra{ii}{jj}$ introduces an additional factor $-1/\kappa$ in the off-diagonal block.

As in the original proof, the spectrum splits into contributions from ($\{\ket{ij}\}_{i\not=j}$), yielding multiplicities $m_{i}\cdot (M-1)$ for each $a_{i}^2$; and the $r$-dimensional subspace spanned by $\{\ket{ii}\}_{i=1}^{r}$, where the problem is reduced to the matrix $A$ defined above. Its eigenvalues, following \cite{jivulescu_positive_2015} are precisely $b_{i}^2$ with multiplicities $m_{i}-1$ and the simple roots $\eta_{i}$ for the defined function $F_{x}(y)$, which only differs from the literature in the $1/\kappa$ multiplicative factor.

Interlacing and the explicit formula for $\eta_{q}$ are also analogous to \cite{jivulescu_positive_2015}. Since $\Tr(A)\not=0$, some extra terms appear in the expression of $\eta_{q}$. Moreover, notice that given $\kappa>1$, $\eta_{q}$ can be positive, negative and $0$.
\end{proof}
\begin{proposition}
\label{Prop:Our6.1.}
Let $\ket{\psi}$ be a normalized pure state in $\mathbb{C}^{N}\otimes\mathbb{C}^{M}$ with maximal Schmidt rank $N$ and let $\alpha\in[-1,\frac{N\cdot(M\kappa-1)}{N-\kappa}]$. Then,
\begin{equation}
\label{eq:PropOur6.1.}
\Lambda_{\alpha}(\rho)=\mathds{1}+\alpha\cdot\ketbra{\psi}{\psi} \in \text{PREDFS}_{\kappa}
\end{equation}
\end{proposition}
\begin{proof}
This proof is analogous to the one of Proposition 6.1. (1) of \cite{jivulescu_positive_2015} renaming the unnormalized isotropic states as the action of the reduction map on the whole pure state and considering $\mathds{1}\otimes\text{RED}_{\kappa}$ with values of $\kappa>1$. 

We first consider a pure state of arbitrary Schmidt rank $r$. For $r\leq\kappa$, positivity follows from $\kappa$-positivity of the reduction map. For $r>\kappa$, Theorem~\ref{thm:Our3.1.} bounds the unique negative eigenvalue with equality for the uniform Schmidt vector. Adding the identity contribution gives positivity whenever $\alpha\leq r(M\kappa-1)/(r-\kappa)$. The biggest possible value of $\alpha$ is then given when $r=N$.
\end{proof}
\end{document}